\documentclass[lettersize,journal]{IEEEtran} 
\usepackage{amsmath,amsfonts}
\usepackage{algorithmic}
\usepackage{algorithm}
\usepackage{subfigure}
\usepackage{array}
\usepackage[caption=false,font=normalsize,labelfont=sf,textfont=sf]{subfig}
\usepackage{textcomp}
\usepackage{stfloats}
\usepackage{url}
\usepackage{verbatim}
\usepackage{graphicx}
\usepackage{cite}
\usepackage{amsthm}
\usepackage{multirow}
\usepackage{color}

\usepackage{mathrsfs} 
\usepackage{amssymb}
\usepackage{bm}
\usepackage{makecell}
\usepackage{etoolbox}

\usepackage{stfloats}
\newtheorem{theorem}{Theorem}

\newtheorem{remark}{Remark}
\usepackage{setspace}

\input epsf

\begin{document}
\title{\huge{AFDM-Enabled Integrated Sensing and Communication: Theoretical Framework and Pilot Design}}
 
\author{ 
 {Fan~Zhang,  
 Zhaocheng~Wang,~\IEEEmembership{Fellow,~IEEE},
 Tianqi~Mao,~\IEEEmembership{Member,~IEEE}, 
  Tianyu~Jiao, Yinxiao Zhuo,
 Miaowen~Wen,~\IEEEmembership{Senior Member,~IEEE}, 
 Wei~Xiang,~\IEEEmembership{Senior Member,~IEEE},
 Sheng~Chen,~\IEEEmembership{Life Fellow,~IEEE}, \\and George K. Karagiannidis,~\IEEEmembership{Fellow,~IEEE} }
\thanks{This work was supported in part by National Key R$\&$D Program of China under Grant 2021YFA0716600,
in part by National Natural Science Foundation of China under Grants 62401054 and 62088101, in part by the Open Project Program of State Key Laboratory of CNS/ATM under Grant 2024B12, and in part by Young Elite Scientists Sponsorship Program by CAST under Grant 2022QNRC001.
\emph{(Fan Zhang and Zhaocheng Wang are Co-first authors with equal contribution.)} \emph{(Corresponding author: Tianqi Mao.)}} 
\thanks{F.~Zhang, Z.~Wang, T.~Jiao, and Y.~Zhuo are with Department of Electronic Engineering, Tsinghua University, Beijing 100084, China (e-mails: zf22@mails.tsinghua.edu.cn, 
zcwang@tsinghua.edu.cn, jiaoty22@mails.tsinghua.edu.cn, zhuoyx20@mails.tsinghua.edu.cn).} 
\thanks{T. Mao is with State Key Laboratory of Environment Characteristics and Effects for Near-space, Beijing Institute of Technology, Beijing 100081, China, and is also with Greater Bay Area Innovation Research Institute of BIT, Zhuhai 519000, China (e-mail: maotq@bit.edu.cn).}
\thanks{M. Wen is with the School of Electronic and Information Engineering, South China University of Technology, Guangzhou, China (e-mail: eemwwen@scut.edu.cn).} %
\thanks{W. Xiang is with the School of Computing Engineering and Mathematical Sciences, La Trobe University, Melbourne, VIC 3086, Australia (e-mail:w.xiang@latrobe.edu.au).} %
\thanks{S. Chen is with the School of Computer Science and Technology, Ocean University of china, Qingdao 266100, China (e-mail: sqc@ecs.soton.ac.uk).} %
\thanks{G. K. Karagiannidis is with Department of Electrical and Computer Engineering, Aristotle University of Thessaloniki, Greece (e-mail: geokarag@auth.gr).} %
\vspace*{-8mm}
} %

\maketitle

\begin{abstract}
The integrated sensing and communication (ISAC) has been envisioned as one representative usage scenario of sixth-generation (6G) network. 
However, the unprecedented characteristics of 6G, especially the doubly dispersive channel, make classical ISAC waveforms rather challenging to guarantee a desirable performance level.
The recently proposed affine frequency division multiplexing (AFDM) can attain full diversity even under doubly dispersive effects, thus becoming a competitive candidate for next-generation ISAC waveforms. 
Relevant investigations are still at an early stage, which involve only straightforward design lacking explicit theoretical analysis. 
This paper provides an in-depth investigation on AFDM waveform design for ISAC applications. Specifically, the closed-form Cr\'{a}mer-Rao bounds of target detection for AFDM are derived, followed by a demonstration on its merits over existing counterparts.
Furthermore, we formulate the ambiguity function of the pilot-assisted AFDM waveform for the first time, revealing conditions for stable sensing performance. 
To further enhance both the communication and sensing performance of the AFDM waveform, we propose a novel pilot design by exploiting the characteristics of AFDM signals.
The proposed design is analytically validated to be capable of optimizing the ambiguity function property and channel estimation accuracy simultaneously as well as overcoming the sensing and channel estimation range limitation originated from the pilot spacing. Numerical results have verified the superiority of the proposed pilot design in terms of dual-functional performance.
\end{abstract}

\begin{IEEEkeywords}
Integrated sensing and communication (ISAC), sixth generation (6G), affine frequency division multiplexing (AFDM), Cr\'{a}mer-Rao bound (CRB), waveform design.
\end{IEEEkeywords}

\section{Introduction}\label{S1}

To support a plethora of emerging applications, such as autonomous driving, Internet of Things (IoT), and low-altitude economy, next-generation wireless networks are expected to not only maintain high-speed connectivity but also achieve reliable and high-resolution sensing capabilities \cite{Liu_jsac_22, Cui_network_21, Dong_twc_23}.  
To enable radar sensing with existing infrastructure of wireless communication networks, the integrated sensing and communication (ISAC) philosophy has been proposed as a promising solution for the sixth-generation (6G) networks \cite{Gonz_proc_24,Saad_Network_20}.
By sharing identical hardware and spectrum resources for both communication and sensing purposes, ISAC technology can realize dual functions simultaneously with desirable performance trade-off and efficient resource usage \cite{Zhou_ojcoms_22, Liu_spm_23, Li_tvt_24}.
Furthermore, it can facilitate information exchange between sensing and communication modules, thereby further improving the performance of dual functions \cite{Fan_iwcmc_24, Kaushik_csm_24, Zhang_icc_24}.

The integrated waveform design that optimally balances both communication and sensing performance is a core challenge in ISAC systems \cite{Sturm_procieee_11, Koivunen_spm_24}. 
Orthogonal frequency division multiplexing (OFDM), which is adopted in most commercial communication standards, has been extensively investigated within the context of ISAC waveform design \cite{Fan_jsac_24,Shi_systemsJ_21,Chen_cl_23}. Specifically, OFDM allows for flexible allocation of orthogonal resource elements (REs) between communication and sensing without mutual interference to achieve a good performance trade-off in time invariant channels \cite{Fan_wcnc_24}. 
However, in the doubly dispersive channel, which is prevalent when employing millimeter-wave/Terahertz frequencies under high mobility, OFDM suffers from notorious inter-carrier interference (ICI), causing communication performance degradation \cite{Liu_twc_15}. 
Furthermore, the destroyed subcarrier orthogonality incurs substantial mutual interference between communication and sensing, further impairing the performance of both functions.

To address this issue, several novel waveforms have been developed to ensure superior ISAC performance \cite{OTFS1, OTFS2, OCDM1, OCDM2}. The work \cite{OCDM1} designed the multicarrier orthogonal chirp division multiplexing (OCDM) by using a series of chirp signals, which can achieve full diversity in frequency-selective channels, thereby showing better gain over OFDM \cite{OCDM2}. 
However, this scheme fails to sufficiently exploit the diversity gain in doubly dispersive channels, and it is fragile to fast time variations. Orthogonal time frequency space (OTFS) \cite{OTFS1}, which is a two-dimensional modulation technique that modulates communication symbols in the delay-Doppler domain, takes advantage of the diversity of both time and frequency shifts caused by high mobility, to ensure robust communication and sensing performance in doubly dispersive channels. Nevertheless, the large pilot overhead and the complex transceiver algorithms hinder OTFS from practical implementation.

Recently, affine frequency division multiplexing (AFDM) was proposed as a generalized form of OFDM and OCDM \cite{AFDM}. AFDM is based on the discrete affine Fourier transform (DAFT), and the parameters of DAFT can be optimized according to channel characteristics to achieve full diversity gain under doubly dispersive channels. Compared to existing waveforms, AFDM can deliver superior communication performance under high mobility scenarios with reduced complexity and channel estimation overhead, making it a promising candidate for next-generation ISAC applications. 

\subsection{Related Works on AFDM-Enabled ISAC}\label{S1.1}
Conventional OFDM-based waveform designs are typically developed under the assumption of small Doppler shifts, where inter-carrier interference is negligible. This allows orthogonal resource elements to be allocated to communication and sensing functions separately, thereby avoiding mutual interference. In contrast, AFDM inherently exhibits coupling effect among subcarriers. The multipath propagation causes the energy of one subcarrier to spread across others, making it difficult to isolate the symbol carried by each subcarrier. As a result, conventional OFDM-based waveform design methodologies are no longer feasible in AFDM systems.

Compared to OFDM, the AFDM waveform has two adjustable parameters, $c_1$ and $c_2$, which increase the flexibility of ISAC waveform design. 
Specifically, the impact of $c_1$ on the bit error rate (BER) and the Cr\'{a}mer-Rao bound (CRB) of target ranging was investigated in  \cite{afdm_isac0}, and an adaptive waveform design method was proposed to achieve different communication and sensing requirements by flexible adjustment of $c_1$. However, this study only offers numerical results and lacked theoretical analysis to support its conclusions.
The work \cite{afdm_isac1} derived a simplified form of the AFDM ambiguity function. Based on this derivation, a parameter selection criterion was proposed to optimize the delay estimation resolution without considering the Doppler shift aspect.

The study \cite{afdm_isac2} proposed an AFDM-enabled ISAC system and designed two algorithms for target distance and velocity estimation in the time and DAFT domains, respectively. This system design employs entirely random communication symbols for sensing purpose, which fails to guarantee a stable sensing performance.
To tackle this issue, a pilot-assisted ISAC scheme was presented in \cite{afdm_isac3}, where one single pilot symbol was utilized for target sensing rather than the whole frame. In addition, a low-complexity self-interference cancellation method was proposed to improve the estimation accuracy. However, to avoid mutual interference, a sufficient number of zero-padded guard subcarriers have to be inserted between the pilot and data, which wastes spectrum resources.

To enhance spectral efficiency, the superimposed pilot philosophy was proposed in \cite{super_pilot}, where pilots and data symbols are superimposed together to activate all subcarriers for transmission. However, the pilot spacing strongly depends on the maximum channel delay, and estimation performance is poor for channels with large delay spread. The sensing range of this design is also limited by the pilot spacing, making it difficult to satisfy the requirement of a large sensing range.

\subsection{Motivation and Our Contributions}\label{S1.2}

It can be seen from the above literature review that investigations on AFDM-enabled ISAC are still at a nascent stage, especially lacking theoretical analysis for its sensing performance. Existing AFDM-enabled ISAC waveforms struggle to balance the spectral efficiency, channel estimation accuracy and sensing capability simultaneously. Against this background, for the first time we derive the CRBs of distance and velocity estimation for AFDM-enabled sensing, along with the two-dimensional ambiguity function, thereby refining the theoretical framework for AFDM-enabled ISAC systems.
Furthermore, by considering the superimposed pilot architecture, we propose a novel pilot design to optimize the ambiguity function property and the channel estimation performance simultaneously. The main contributions of this paper are summarized as follows.

\begin{enumerate}
\item To evaluate the sensing capability of AFDM-enabled ISAC waveform, closed-form CRB expressions for the distance and velocity estimation are derived. These results are applicable to other state-of-the-art multi-carrier waveforms that are special cases of AFDM, i.e., OFDM and OCDM, by setting $c_1\! =\! 0$ and $c_1\! =\! 1/2N$, respectively, where $N$ is the number of DAFT points. Based on the derived CRBs and numerical analysis, we highlight the significant advantages of AFDM over existing counterparts, in terms of ISAC applications.

\item We present the formulation of the ambiguity function for the pilot-assisted AFDM waveform and analyze its statistical properties. Through theoretical analysis, we derive two conditions that can minimize the variance of the ambiguity function to realize stable sensing performance, which offer inspiring guidelines for AFDM-enabled ISAC system design.

\item To further enhance the dual-functional performance of the pilot-assisted AFDM waveform, a novel pilot design is proposed, which incorporates the parameter $c_1$ setting, the pilot spacing configuration and the pilot sequence design. The pilot sequence is designed to achieve an ideal ambiguity function by leveraging the unique multipath separation characteristics of AFDM based on the constant amplitude zero autocorrelation (CAZAC) sequences. The proposed design is validated theoretically to realize minimal channel estimation error.
More importantly, it can overcome the limitations on maximum sensing range and delay spread imposed by the pilot spacing, and exhibits good dual-functional performance under large propagation delays.

\item Numerical results are offered to validate the superiority of the proposed pilot design in terms of communication and sensing performance. Compared to existing pilot schemes, our pilot design provides more accurate channel estimation, resulting in a lower BER at the same signal-to-noise ratio (SNR). Furthermore, owing to its ideal ambiguity function, our pilot design does not introduce high sidelobe interference in the target sensing, significantly reducing the false alarm probability and leading to a preferable receiver operating characteristic (ROC) curve. Simulations also show that the velocity and distance estimation errors of the proposed design are close to the derived CRBs.
\end{enumerate}

\subsection{Organization and Notations}\label{S1.3}

The remainder of this paper is organized as follows. Section~\ref{S2} introduces the AFDM-enabled ISAC system model. Section~\ref{S3} presents our theoretical framework for AFDM-enabled ISAC by deriving the sensing CRBs and the ambiguity function of the pilot-assisted AFDM waveform. In Section~\ref{S4}, a novel pilot design is proposed to enhance the dual-functional performance. Numerical results are provided in Section~\ref{S5}, followed by concluding remarks drawn in Section~\ref{S6}.

\emph{Notation:} 
The $m$-th element of vector $\mathbf{x}$ is denoted by $x[m]$, while $X[m,k]$ denotes the element in the $m$-th row and the $k$-th column of matrix $\mathbf{X}$. $(\cdot)^*$, $(\cdot)^{\rm T}$, and $(\cdot)^{\rm H}$ denote the conjugate, transpose, and Hermitian operations, respectively. $\mathbf{I}_N$ denotes the $N \times N$ identity matrix. $\mathbf{x}\! \sim\! \mathcal{CN}(\mathbf{0},\sigma^2\mathbf{I}_N)$ represents that each element of $\mathbf{x}\! \in\! \mathbb{C}^{N\times 1}$ follows a complex Gaussian distribution with mean $0$ and variance $\sigma^2$. The operator $\text{diag}(\mathbf{x})$ transforms vector $\mathbf{x}$ into a diagonal matrix.
$\|\mathbf{X}\|$ denotes the Frobenius norm of matrix $\mathbf{X}$.
$\lfloor\cdot \rfloor$ is the floor function that rounds a real number to the nearest integer less than or equal to that number. $|\cdot |$ denotes the modulus operation. $\mathbb{E}\{\cdot\}$ denotes the expectation operation. 
$\langle \cdot \rangle_N$ denotes the modulo operation with divisor $N$. $\Re \{\cdot\}$ denotes the real part of a complex number.

\begin{figure}[!h]
\vspace*{-3mm}
\center
\includegraphics[width=1\linewidth, keepaspectratio]{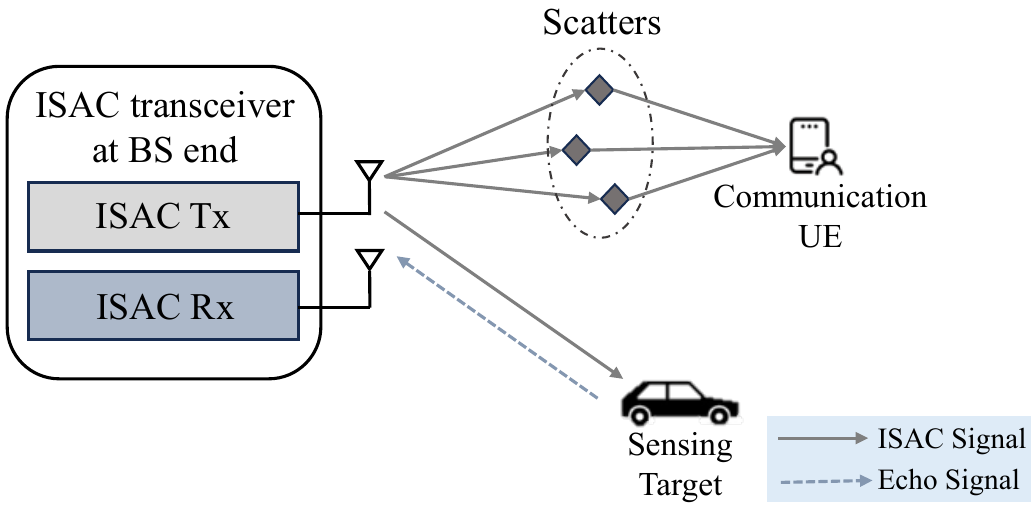}	
\vspace*{-6mm}
\caption{Typical monostatic ISAC scenario.}
\label{fig1}
\vspace*{-2mm}
\end{figure} 

\section{System Model}\label{S2}
We consider a typical monostatic ISAC scenario as illustrated in Fig.~\ref{fig1}, where the base station (BS) simultaneously transmits data to user equipment (UE) while leveraging the received echo signals to detect the presence of targets in the surrounding environment. The communication channel refers to the direct link between the BS and the UE, which consists of multiple clustered paths formed by scatterers between the base station and the user. The sensing channel represents the target reflection path observed at the BS, which is composed of reflection paths generated by the target. These two channels are independent and uncorrelated in our considered model.Additionally, when a target is detected, the BS estimates its distance and relative radial velocity.
Perfect cancellation of the resultant self-interference is assumed with sufficient isolation between the ISAC transmitter (Tx) and receiver (Rx), achieved by active radio-frequency cancellation methods or other related techniques \cite{STAR}.

\subsection{Signal Model}\label{S2.1}

The superimposed pilot scheme \cite{super_pilot} is considered for the AFDM-enabled ISAC system owing to its high spectral efficiency. The $N_c$-length transmit vector in the DAFT domain is expressed as the superposition of the pilot and data parts:
\begin{equation}\label{signal} 
  \mathbf{x} = \mathbf{x}_{\rm{p}}+\mathbf{x}_{\rm{d}},
\end{equation}
where $\mathbf{x}_{\rm{p}}\! \in\! \mathbb{C}^{N_c\times 1}$ and $\mathbf{x}_{\rm{d}}\! \in\! \mathbb{C}^{N_c\times 1}$ denote the pilot and data signals, respectively. The pilot signal is deterministic with power $\sigma_{\rm{p}}^2\! =\! \mathbf{x}_{\rm{p}}^{\rm H}\mathbf{x}_{\rm{p}}$, while the data counterpart is stochastic with $\mathbb{E}\big\{\mathbf{x}_{\rm{d}}\mathbf{x}_{\rm{d}}^{\rm H}\big\}\! =\! \sigma_{\rm{d}}^2\mathbf{I}_{N_{c}}$. The overall signal power is given by 
\begin{equation}\label{total_power} 
  P_{\rm{t}} = \sigma_{\rm{p}}^2 + N_c\sigma_{\rm{d}}^2.
\end{equation}
By performing the inverse DAFT on $\mathbf{x}$, the transmit signal in the time domain can be written as \cite{afdm_signal}
\begin{align}\label{single} 
	s[n] \!=& \frac{1}{\sqrt{N_c}} \!\!\sum_{m=0}^{N_c-1}\!x[m]e^{\textsf{j}2\pi\eta_m[n]}, ~n = 0,\cdots,N_c-1 ,
\end{align}
with 
\begin{align}\label{eq-eta} 
	\eta_m[n] = c_1 n^2 + \frac{1}{N_c}mn + c_2 m^2,
\end{align}
where $c_1$ and $c_2$ are the DAFT parameters. To guarantee the optimal diversity order for AFDM, $c_2$ should be an irrational number, and $c_1$ should satisfy
\begin{align}\label{eq-c1} 
  \frac{2\nu_{\rm{m}}+1}{2N_c} \leq c_1 \leq \frac{N_c}{2N_c(\tau_{\rm{m}}+1)}, 
\end{align}
where $\tau_{\rm{m}}$ and $ \nu_{\rm{m}}$ denote the maximum normalized delay and Doppler shift, respectively. The lower bound is set to ensure that the effective channels corresponding to different paths do not overlap in the DAFT domain, while the upper bound is introduced to constrain the range of the maximum delay to avoid ambiguity in the DAFT representation of the effective channel. The process in (\ref{single}) can be written in matrix form as
\begin{align}\label{matrix_form} 
  \mathbf{s} =& \mathbf{\Lambda}^{\rm H}_{c_2} \mathbf{F}^{\rm H} \mathbf{\Lambda}^{\rm H}_{c_1}\mathbf{x} \triangleq \mathbf{A}^{\rm H}\mathbf{x},
\end{align}
where $\mathbf{s}\! =\! \left[s[0],s[1],\cdots,s[N_c-1]\right]^{\rm T}\! \in\! \mathbb{C}^{N_c\times 1}$, $\mathbf{F}$ is the $N_c$-point discrete Fourier transform (DFT) matrix, and $\mathbf{\Lambda}_{c_i}\! =\! \text{diag}\bigg(\Big[1,e^{\textsf{j}2\pi c_i},\cdots,e^{\textsf{j}2\pi c_i (N_c-1)^2}\Big]^{\rm T}\bigg) \! \in\! \mathbb{C}^{N_c\times N_c}$, while $\mathbf{A}\! =\! \mathbf{\Lambda}_{c_1}\mathbf{F}\mathbf{\Lambda}_{c_2}\! \in\! \mathbb{C}^{N_c\times N_c}$.

\begin{figure}[!b]
\vspace*{-4mm}
\center
\includegraphics[width=0.8\linewidth, keepaspectratio]{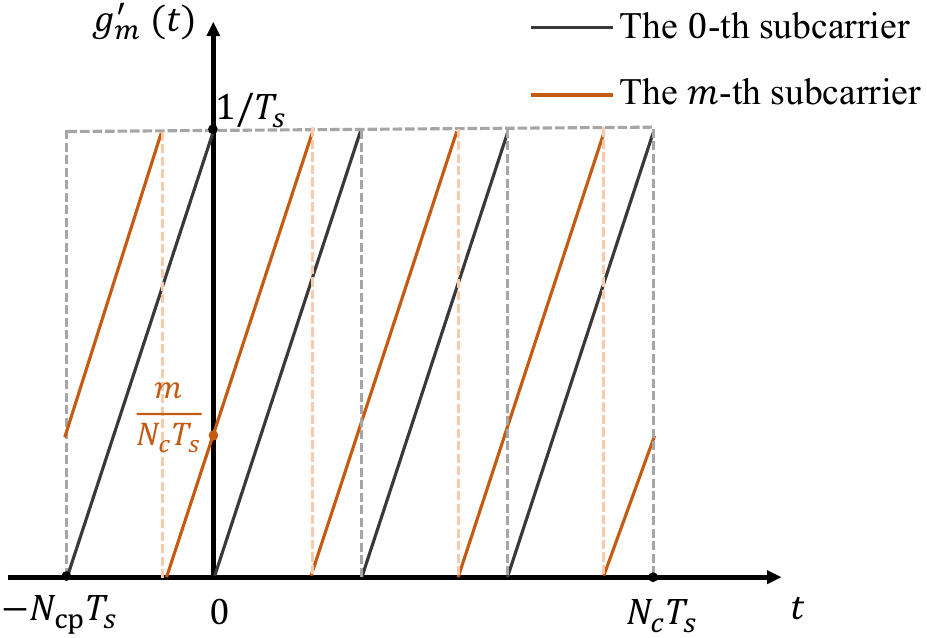}	
\vspace*{-3mm}
\caption{The frequency of the AFDM subcarriers with respect to time.}
\label{fig2}
\vspace*{-1mm}
\end{figure} 

To combat the inter-symbol interference (ISI) incurred by multipath propagation, a chirp-periodic prefix (CPP) is inserted in front of each AFDM symbol. The length of the CPP, denoted as $N_{\rm{cp}}$, is chosen to be larger than both the maximum integer delay of the communication channel and the maximum integer round trip delay of the sensing target. According to the periodicity of AFDM, the CPP can be expressed as
\begin{align}\label{cpp} 
	s[n] \!=& s[N_c + n] e^{-\textsf{j}2\pi c_1( N_c^2 +  2N_cn )}, \, n\! = \!-N_{\rm{cp}},\!\cdots\!,-\!1.
\end{align}
After adding the CPP, the discrete signals are fed into a digital-to-analog converter (D/A) to obtain the continuous transmit signal as
\begin{align}\label{continous} 
	s(t)\!=& \frac{1}{\sqrt{N_{c}}} \!\sum_{m=0}^{N_{c}\!-1} \! x[m] e^{\mathrm{j} 2 \pi\left( g_m(t-N_{\rm{cp}}T_{s}) + c_{2} m^{2}\right)},  
\end{align}
where  $t\! \in\! \Big[ 0, \, \left(N_{c} + N_{\rm{cp}}\right)T_{s} \Big)$ and $T_s$ is the sampling period, while $g_m(t)$ represents the instantaneous phase of the $m$-th subcarrier, and its derivative $g'_m(t)$ represents the instantaneous frequency as illustrated in Fig.~\ref{fig2}.
Due to the frequency wrapping property of AFDM chirp subcarriers, $g_m(t)$ can be expressed as
\begin{equation}\label{eq-Ph} 
  g_m(t) = c_{1} \left(\frac{t}{T_{s}}\right)^{2}\!+\frac{t}{N_{c}T_{s}} m
  - q_m\left(\frac{t}{T_s}\right)\frac{t}{T_s} ,
\end{equation}
where $q_m\left(\frac{t}{T_s}\right) = \big\lfloor 2c_1\frac{t}{T_s} + \frac{m}{N_c} \big\rfloor$.

\begin{figure*}[!t]
\vspace*{-1mm}
\center
\includegraphics[width=0.9\linewidth, keepaspectratio]{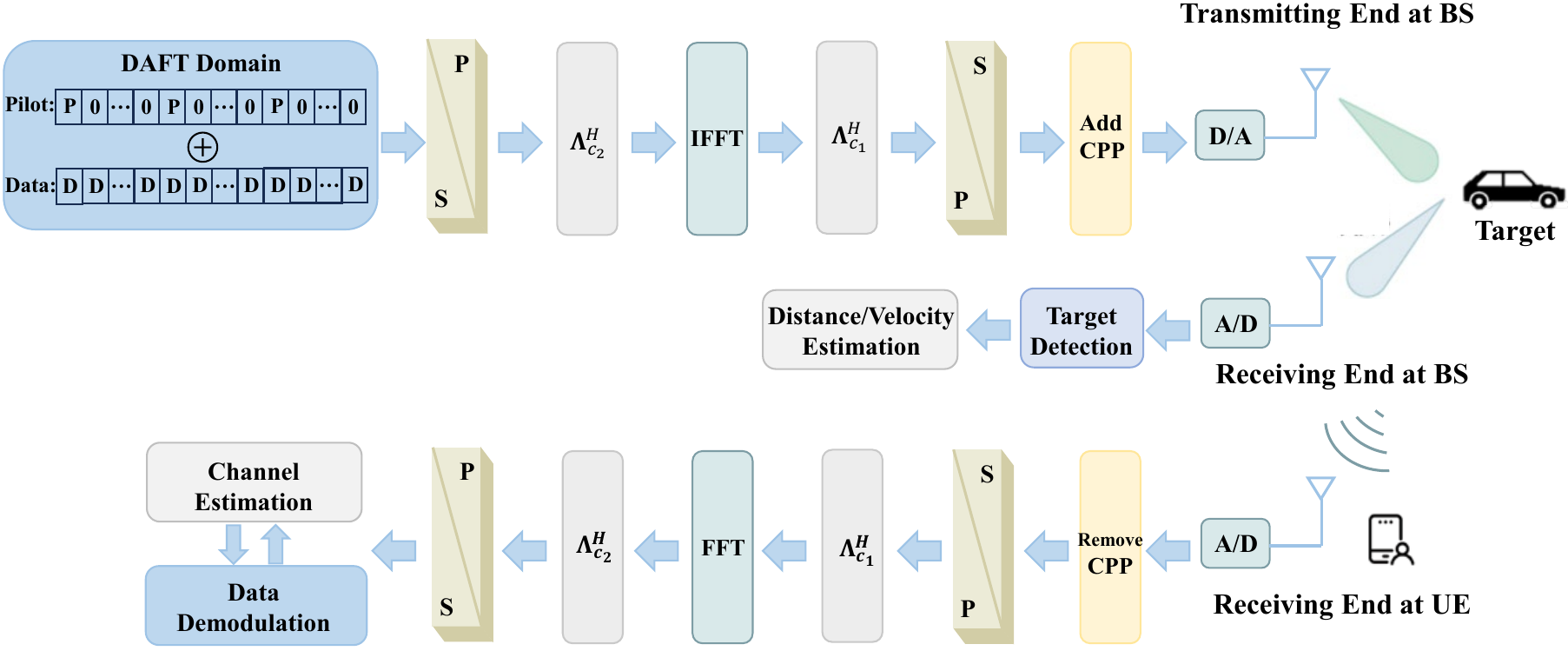}	
 \vspace*{-2mm}
\caption{Transceiver model of the AFDM-enabled ISAC system, where the signal reflected by the target is received by the receiving end at the BS for sensing purposes, and the signal through the downlink communication channel is received by the receiving end at the UE.} 
\label{fig3}
\vspace*{-3mm}
\end{figure*}

\subsection{Channel Estimation and Data Demodulation}\label{S2.2}

The multi-path time-varying communication channel is considered. By denoting the complex gain, delay and Doppler shift of the $i$-th path as $\alpha_{i}$, $\mu_i$ and $f_i$, respectively, the channel coefficient of the $i$-th path can be expressed as
\begin{align}\label{eqh} 
	h_{i}(t, \mu) = \alpha_{i}e^{\textsf{j}2\pi f_it}\delta(\mu - \mu_i).
\end{align}
After the analog-to-digital conversion, discarding CPP and performing $N_c$-point DAFT, the effective channel matrix of the $i$-th path in the DAFT domain can be written as \cite{AFDM}
\begin{align}\label{Hi} 
  \mathbf{H}_{i} =& \alpha_{i} \mathbf{A}\mathbf{\Gamma}_{\text{CPP}_i}\mathbf{\Pi}^{\tau_i}\mathbf{\Delta}_{\nu_i}\mathbf{A}^{\rm H},
\end{align} 
where $\tau_i\! =\! \lfloor \mu_i/T_s\rfloor\! \in\! [0,\, \tau_{\rm{m}}]$ denotes the normalized delay and $\nu_i\! =\! f_iN_cT_s\! =\! f_i / \Delta f\! \in\! [-\nu_{\rm{m}},\, \nu_{\rm{m}}]$ denotes the normalized Doppler shift, while $\Delta f\! =\! 1/N_cT_s$ is the subcarrier spacing. 
Besides, $\mathbf{\Gamma}_{\text{CPP}_i}\! \in\! \mathbb{C}^{N_c \times N_c}$ is the CPP matrix given by
\begin{align}\label{eq-Gamma} 
  \mathbf{\Gamma}_{\text{CPP}_i} =& \text{diag}\Bigg(
    \begin{cases}
      e^{-\textsf{j}2\pi c_1( N_c^2 +  2N_c(n-\tau_i) )}, ~  n < \tau_i , \\
      1, \quad n \geq \tau_i ,
    \end{cases}\!\!\!\! \Bigg),  
\end{align}
where $n$ is the index of the diagonal elements, $\mathbf{\Pi}$ is the forward cyclic-shift matrix written as
\begin{align}\label{eq-Pi} 
  \mathbf{\Pi} = \left[
    \begin{array}{cccc}
      0 & \cdots & 0  & 1\\
      1 & \cdots & 0  & 0\\
      \vdots & \ddots & \vdots  & \vdots\\  
      0 & \cdots & 1  & 0
    \end{array}
  \right],
\end{align}
and $\mathbf{\Delta}_{\nu_i}\! =\! \text{diag}\big(1,e^{\textsf{j}2\pi \nu_i/N_c},\cdots,e^{\textsf{j}2\pi \nu_i (N_c-1)/N_c}\big)$ stands for the Doppler shift matrix.

Since the channel matrix for the path with a fractional normalized Doppler shift can be represented as a linear combination of matrices for paths with integer normalized Doppler shift, we use all integer Doppler paths as the basis to express the communication channel. Considering the maximum delay and Doppler shift, there are at most $L_{\rm{m}}\! =\! (2\nu_{\rm{m}}+1)(\tau_{\rm{m}}+1)$ paths. Therefore, the received signal $\mathbf{y}\! \in\! \mathbb{C}^{N_c\times 1}$ in the DAFT domain can be expressed as 
\begin{align}\label{r1} 
  \mathbf{y} =& \mathbf{H}_{\rm{eff}}\mathbf{x} + \mathbf{w}_{\rm{c}} = \sum_{i=1}^{L_{\rm{m}}} \mathbf{H}_{i}\mathbf{x} + \mathbf{w}_{\rm{c}} \nonumber \\
  =& \sum_{i=1}^{L_{\rm{m}}}\alpha_i \mathbf{\Phi}_i(\mathbf{x}_{\rm{p}}+\mathbf{x}_{\rm{d}}) + \mathbf{w}_{\rm{c}}, 
\end{align}
where $\mathbf{w}_{\rm{c}}\! \sim\! \mathcal{CN}(\mathbf{0},\sigma_{\rm{cn}}^2\mathbf{I}_N)$  represents the noise, and $\mathbf{\Phi}_i\! =\! \mathbf{A}\mathbf{\Gamma}_{\text{CPP}_i}\mathbf{\Pi}^{\tau_i}\mathbf{\Delta}_{\nu_i}\mathbf{A}^{\rm H}$ with $\tau_i$ and $\nu_i$ given by 
\begin{align} 
  \tau_i =& \left\lfloor \frac{i-1}{2\nu_{\rm{m}}+1} \right \rfloor, \label{eq-tau} \\
   \nu_i =& \langle{i-1}\rangle_{2\nu_{\rm{m}}+1} - \nu_{\rm{m}}. \label{eq-nu}
\end{align}  

To estimate the effective channel $\mathbf{H}_{\rm{eff}} $, we employ the pilot signal to obtain the complex gain of each possible path.
By denoting $\bm{\alpha}\! =\!\big[\alpha_1,\alpha_2, \cdots,\alpha_{L_{\rm{m}}}\big]^{\rm T}$, (\ref{r1}) can be rewritten as
\begin{align}\label{r2} 
  \mathbf{y} =& \big(\mathbf{\Psi}_{\rm{p}}+\mathbf{\Psi}_{\rm{d}}\big)\bm{\alpha} + {\mathbf{w}}_{\rm{c}} =\mathbf{\Psi}_{\rm{p}} \bm{\alpha} + \hat{\mathbf{w}},
\end{align}
where $\mathbf{\Psi}_{\rm{p}}\! =\! \big[\mathbf{\Phi}_1\mathbf{x}_{\rm{p}}, \mathbf{\Phi}_2\mathbf{x}_{\rm{p}}, \cdots, \mathbf{\Phi}_ {L_{\rm{m}}}\mathbf{x}_{\rm{p}}\big]\! \in\! \mathbb{C}^{N_c \times L_{\rm{m}} }$ and $\mathbf{\Psi}_{\rm{d}}\! =\! \big[\mathbf{\Phi}_1\mathbf{x}_{\rm{d}}, \mathbf{\Phi}_2\mathbf{x}_{\rm{d}}, \cdots, \mathbf{\Phi_ {L_{\rm{m}}}}\mathbf{x}_{\rm{d}}\big]\! \in\! \mathbb{C}^{N_c \times L_{\rm{m}}}$, while $\hat{\mathbf{w}} \triangleq \mathbf{\Psi}_{\rm{d}} \bm{\alpha} + \mathbf{w}_{\rm{c}}$ defines the effective noise with the covariance matrix $\mathbf{C}_{\hat{\mathbf{w}}}$. 
Using the minimum mean square error (MMSE) algorithm yields the estimated complex gain vector: 
\begin{align}\label{eq-MMSE} 
  \hat{\bm{\alpha}} =& \left(\mathbf{\Psi}_{\rm{p}}^{\rm H}\mathbf{C}_{\hat{\mathbf{w}}}^{-1}\mathbf{\Psi}_{\rm{p}} + \mathbf{C}_{\bm{\alpha}}^{-1} \right)^{-1} \mathbf{\Psi}_{\rm{p}}^{\rm H}\mathbf{C}_{\hat{\mathbf{w}}}^{-1}\mathbf{y}, 
\end{align}
where $\mathbf{C}_{\bm{\alpha}}\! =\! E\big\{\bm{\alpha}\bm{\alpha}^{\rm H}\big\}\! =\! \text{diag}\Big(\big[\sigma_{\alpha_1}^2, \sigma_{\alpha_2}^2, \cdots, \sigma_{\alpha_{L_{\rm{m}}}}^2\big]^{\rm T}\Big)$ is the covariance matrix of $\bm{\alpha}$.
We define a path indicator vector $\mathbf{b}\! =\! \big[b_1, b_2, \cdots, b_{L_{\rm{m}}}\big]^{\rm T}$, where $b_i$ indicates whether the $i$-th path exists or not, i.e., 
\begin{align}\label{eq-IndC} 
  b_i = \begin{cases}
    1 ,~\text{if} , \hat{\alpha}_i > \epsilon , \\
    0 ,  ~\text{otherwise}, 
  \end{cases}
\end{align}
in which $\epsilon$ is a predefined threshold. Then the estimated effective channel is given by 
\begin{align}\label{eq-CEf} 
  \hat{\mathbf{H}}_{\rm{eff}} =& \sum_{i=1}^{L_{\rm{m}}} b_i \hat{\alpha}_i \mathbf{\Phi}_i.
\end{align}
The communication data are demodulated using the estimated channel. The channel estimation and data demodulation can be iteratively operated to reduce the mutual interference between the data and pilot signals, thereby enhancing the channel estimation accuracy and data detection performance \cite{super_pilot}.

\subsection{Target Sensing}\label{S2.3}

Assume that the transmit signal is reflected by a point target with the distance $R$ and relative radial velocity $V$. The received echoes at the BS can be expressed as
\begin{equation}\label{eq-Echo} 
  r_{\rm{s}}(t) = \beta s(t- \bar{\mu}) e^{\textsf{j}2\pi \bar{f}t} + w_{\rm{s}}(t),
\end{equation}
where $\bar{\mu}\! =\! \frac{2R}{c}$ and $\bar{f}\! =\! \frac{2Vf_c}{c}$ with $f_c$ and $c$ being the carrier frequency and the speed of light, respectively, $\beta$ is the complex gain, and $w_{\rm{s}}(t)\! \sim\! \mathcal{CN}(0,\sigma_{\rm{s}}^2)$ denotes the thermal noise plus the clutters.
After passing through the analog-to-digital converter, the received signal is sampled with the period $T_s$, yielding
\begin{equation}\label{e13} 
  r_{\rm{s}}[n] = r_{\rm{s}}(nT_s) = \beta s(nT_s - \bar{\mu})  e^{\textsf{j}2\pi \bar{\nu}n/N_c} + w_{\rm{s}}[nT_s],
\end{equation}
where $\bar{\nu}\! =\! \bar{f}/\Delta f$ is the normalized Doppler shift. After the CPP removal, the target sensing can be performed using detection algorithms, such as cyclic correlation method \cite{correlation}, 2D-FFT \cite{2d_fft} or MUSIC schemes \cite{music}. 
To simultaneously obtain the distance and velocity information of the target, we consider a Doppler compensation-based correlation method, which yields the range-Doppler function (RDF) written as
\begin{align}\label{rdm} 
  E(\tau, \nu) =& \sum_{n=0}^{N_c-1} r_{\rm{s}}^*[n] s[n-\tau] e^{\textsf{j}2\pi {\nu} n/N_c} ,
\end{align}
where $\tau\! \in\! [0,\, \tau_{\rm{m}}]$ is the index of the correlation function and $\nu\! \in\! [-\nu_{\rm{m}},\, \nu_{\rm{m}}]$ is the normalized frequency offset for Doppler compensation.
Based on the RDF, the target can be detected via the hypothesis test. Let the test threshold be $\gamma$. The hypothesis can be expressed as
\begin{align}\label{eqHtest} 
	\text{H}_1:\frac{|E(\tau,\nu)|^2}{N(\tau,\nu)} >\gamma,\quad \text{H}_0:\frac{|E(\tau,\nu)|^2}{N(\tau,\nu)} < \gamma,
\end{align}
where $\text{H}_1$ and $\text{H}_0$ represent the presence and absence of the target, respectively, $N(\tau,\nu)$ denotes the average noise power at point $(\tau,\nu)$, which can be calculated by averaging $|E(\tau,\nu)|^2$ in the whole or local region in the delay-Doppler domain \cite{rdm}.
If the hypothesis $\text{H}_1$ holds true at point $(\hat{\tau},\hat{\nu})$, i.e., $(\hat{\tau},\hat{\nu})$ corresponds to a target, then the distance and velocity of the target can be derived as
\begin{align} 
  \hat{R} =& \frac{c\hat{\tau}T_s}{2}, \quad\hat{\tau} = [0,\Delta \tau, 2\Delta \tau, \cdots, \tau_{\rm{m}}],  \label{eqE-d} \\
 \hat{ V} =&
   \frac{c\hat{\nu}\Delta f}{2 f_{\rm{c}}}, \quad  \hat{\nu} = [-\nu_{\rm{m}},\cdots,-\Delta \nu, 0,\Delta \nu, \cdots,\nu_{\rm{m}}],
\end{align}
where $\Delta \tau$ and $\Delta \nu,$ are the resolutions for delay and Doppler shift estimations, respectively. According to (\ref{rdm}), $\Delta \tau$ can be reduced by oversampling $r_{\rm{s}}[n]$ and $s[n]$, and $\Delta \nu$ can be reduced with a finer-grained Doppler shift compensation.

\section{Theoretical Framework for AFDM-Enabled ISAC}\label{S3}

The schematic diagram of the transceiver for the proposed AFDM-enabled ISAC system is depicted in Fig.~\ref{fig3}. To investigate the sensing capability of the AFDM waveform, we derive its CRBs for distance and velocity estimations, and analyze its advantage over classical OFDM and OCDM.

\begin{figure*}[!t]\setcounter{equation}{32}
\vspace*{-1mm}
\begin{align}\label{fim} 
  \mathbf{F}_{\mathbf{I}} \approx & \frac{2}{\sigma_{\rm{s}}^2} \left[
    \begin{array}{ccc}
    \frac{P_t}{N_c} & 0 & 0 \\
    0 & \beta^2\frac{(2\pi)^2}{N_c}\sum_n\sum_m P_m(\mathcal{F}_m(n\!-\!\bar{\tau}))^2 & \beta^2\frac{(2\pi)^2}{N_c}\sum_n\sum_m P_m( -\mathcal{F}_m(n-\bar{\tau}))(\frac{n}{N_c}) \\
    0 & \beta^2 \frac{(2\pi)^2}{N_c}\sum_n\sum_m P_m( -\mathcal{F}_m(n-\bar{\tau}))(\frac{n}{N_c}) & \beta^2\frac{(2\pi)^2}{N_c}\sum_n\sum_m P_m (\frac{n}{N_c})^2  
\end{array}
  \right] .
\end{align}
\vspace*{-1mm}
\hrulefill
\vspace*{-1mm}
\begin{align} 
  \text{CRB}_{\bar{\tau}} \approx& \frac{\frac{\sigma_{\rm{s}}^2N_c}{8\pi^2\beta^2} \sum_n\sum_m P_m (\frac{n}{N_c})^2 }{ \Big( \sum_n\sum_m P_m(\mathcal{F}_m(n-\bar{\tau}))^2  \Big) \left(\sum_n\sum_m P_m (\frac{n}{N_c})^2 \right) - \left( \sum_n\sum_m P_m( \mathcal{F}_m(n-\bar{\tau}))(\frac{n}{N_c})\right)^2} , \label{crbt} \\ 
  \text{CRB}_{\bar{\nu}} \approx&  \frac{ \frac{\sigma_{\rm{s}}^2N_c}{8\pi^2\beta^2} \sum_n\sum_m P_m(\mathcal{F}_m(n-\bar{\tau}))^2}{\Big( \sum_n\sum_m P_m(\mathcal{F}_m(n-\bar{\tau}))^2 \Big) \left(\sum_n\sum_m P_m (\frac{n}{N_c})^2 \right) - \left( \sum_n\sum_m P_m( \mathcal{F}_m(n-\bar{\tau}))(\frac{n}{N_c})\right)^2}. \label{crbn}
\end{align}
\vspace*{-1mm}
\hrulefill
\vspace*{-3mm}
\end{figure*}

\subsection{Sensing CRBs Analysis}\label{S3.1}

Let $\bm{\theta}\! =\! [\beta,\bar{\tau},\bar{v}]^{\rm T}$ be the unknown parameters for target sensing. According to (\ref{e13}), by denoting $\bar{s}[n]\! =\! \beta s(n T_s- \bar{\tau}T_s) e^{\textsf{j}2\pi \bar{\nu}n/N_c}$, the received echo signal in matrix form can be expressed as\setcounter{equation}{26}
\begin{align}\label{eq-EcoM} 
  \mathbf{r}_{\rm{s}} = \bar{\mathbf{s}} + \mathbf{w}_{\rm{s}},
\end{align}
where $\bar{\mathbf{s}}\! =\! [\bar{s}[0],  \cdots, \bar{s}[N_c\!-\!1]]^{\rm T}$, $\mathbf{r}_{\rm{s}}\! =\! [r_{\rm{s}}[0], \cdots, r_{\rm{s}}[N_c\!-\!1]]^{\rm T}$ and $\mathbf{w}_{\rm{s}}\! =\! [w_{\rm{s}}[0], \cdots, w_{\rm{s}}[N_c\!-\!1]]^{\rm T}\! \sim\! \mathcal{CN}(\mathbf{0},\sigma_{\rm{s}}^2 \mathbf{I}_{N_c})$.
The log-likelihood function of $\mathbf{r}_{\rm{s}}$ with respect to $\bm{\theta}$ is defined by
\begin{align}\label{eq-LLF} 
  \log p(\mathbf{r}_{\rm{s}},\bm{\theta}) = -\frac{1}{\sigma_{\rm{s}}^2}(\mathbf{r}_{\rm{s}} - \bar{\mathbf{s}})^{\rm H}(\mathbf{r}_{\rm{s}} - \bar{\mathbf{s}}) - C_1,
\end{align}
where $C_1\! =\! N_c \log(\sqrt{2\pi}\sigma_{\rm{s}})$ is a constant term. To obtain the Fisher information matrix (FIM), denoted as $\mathbf{F}_{\mathbf{I}}$, we calculate the expectation of the partial derivatives of the log-likelihood function with respect to the elements of $\bm{\theta}$ as follows
\begin{align}\label{e20} 
  F_{I}[i,j] =& -\mathbb{E}\left\{\frac{\partial^2 \log p(\mathbf{r}_{\rm{s}},\bm{\theta})}{\partial \theta_{i}\partial\theta_{j}^{*}}\right\} \nonumber \\
	=& \frac{2}{\sigma_{\rm{s}}^2} \Re \left\{ \mathbb{E}
  \bigg\{ \bigg(\frac{\partial \bar{\mathbf{s}}}{\partial \theta_{j}}\bigg)^{\rm H}\bigg(\frac{\partial \bar{\mathbf{s}}}{\partial \theta_{i}}\bigg) \bigg\}\right\} \nonumber \\
  =& \frac{2}{\sigma_{\rm{s}}^2} \Re \left\{ \mathbb{E}
    \bigg\{ \sum_{n=0}^{N_c-1}
   \left(\frac{\partial {\bar{s}}[n]}{\partial \theta_{j} } \right)^*\left(\frac{\partial {\bar{s}}[n]}{\partial \theta_{i} } \right) \bigg\}
    \right\} ,
\end{align}
where $\theta_{i}=\theta[i]$. The derivatives with respect to $\beta$, $\bar{\nu}$ and $\bar{\tau}$ can be deduced respectively as
\begin{align} 
  \frac{\partial{\bar{s}[n]}}{\partial \beta} =& s(n T_s- \bar{\tau}T_s) e^{\textsf{j} 2\pi \bar{\nu} n/N_c}, \label{e21} \\
  \frac{\partial{\bar{s}[n]}}{\partial \bar{\nu}} =& \beta (\textsf{j} 2\pi n/N_c) s(n T_s - \bar{\tau}T_s) e^{\textsf{j} 2\pi \bar{\nu} n/N_c}, \label{e22} \\
  \frac{\partial{\bar{s}[n]}}{\partial \bar{\tau}} \approx & \frac{\beta e^{\textsf{j}2\pi \bar{\nu}n/N_c}}{\sqrt{N_{c}}} \!\!\sum_{m=0}^{N_{c}-1} \! \Big(\! - \textsf{j} 2 \pi \mathcal{F}_m(n -\bar{\tau}) \!\Big) x[m]  \nonumber \\
  & \times e^{\textsf{j} 2 \pi\left( g_m(nT_s- \bar{\tau}T_s) + c_{2} m^{2}\right)}, \label{e23}
\end{align}
where $\mathcal{F}_m(n-\bar{\tau})\! =\! \big( 2c_1(n\! -\! \bar{\tau})\! +\! \frac{m}{N_c}\big) - \big\lfloor 2c_1(n\! -\! \bar{\tau})\! +\! \frac{m}{N_c}\big\rfloor$ represents the fractional part of $2c_1(n\! -\! \bar{\tau})\! +\! \frac{m}{N_c}$. The derivations of $\frac{\partial{\bar{s}[n]}}{\partial {\beta}}$ and $\frac{\partial{\bar{s}[n]}}{\partial \bar{\nu}}$ are straightforward. The derivation of $\frac{\partial{\bar{s}[n]}}{\partial \bar{\tau}}$ can be found in Appendix~\ref{app0}.

Substituting (\ref{e21})--(\ref{e23}) into (\ref{e20}) yields the approximate FIM for $\bm{\theta}$ (\ref{fim}) given at the top of the next page, where $P_m\! =\! \mathbb{E}\{|x[m]|^2\}$ is the allocated power on the $m$-th subcarrier and $P_t = \sum_{m=0}^{N_c-1}P_m$ is the total power of one AFDM symbol. The derivation of $\mathbf{F}_{\mathbf{I}}$ is given in Appendix~\ref{app0}.

The CRBs of $\bar{\tau}$ and $\bar{\nu}$, denoted as $\text{CRB}_{\bar{\tau}}$ and $\text{CRB}_{\bar{\nu}}$, can be derived by calculating the inverse of $\mathbf{F}_{\mathbf{I}}$, which are given in (\ref{crbt}) and (\ref{crbn}), respectively, at the top of this page.
Specifically, by substituting $P_m\! =\! \mathbb{E}\{|x_{\rm{d}}[m]|^2\}$ and $P_m\! =\! |x_{\rm{p}}[m]|^2 + \mathbb{E}\{|x_{\rm{d}}[m]|^2\}$ into the expression of sensing CRBs, we can obtain the sensing CRBs for full data ADFM signals and pilot-assisted AFDM signals, respectively. Accordingly, the CRBs of distance $R$ and velocity $V$ can be obtained as\setcounter{equation}{35}
\begin{align} 
  \text{CRB}_{R} =& \left(\frac{cT_s}{2}\right)^2\text{CRB}_{\bar{\tau}},  \label{crbr} \\
  \text{CRB}_{V} =& \left( \frac{c\Delta f}{2f_c}\right)^2\text{CRB}_{\bar{\nu}}. \label{crbv}
\end{align}
The above expressions for the CRBs are applicable not only to AFDM but also to OFDM and OCDM by choosing $c_1\! =\! 0$ and $c_1\! =\! \frac{1}{2N_c}$, respectively, in (\ref{crbr}) and (\ref{crbv}). 

\begin{figure}[!b]
\vspace*{-5mm}
\center
\includegraphics[width=1\linewidth, keepaspectratio]{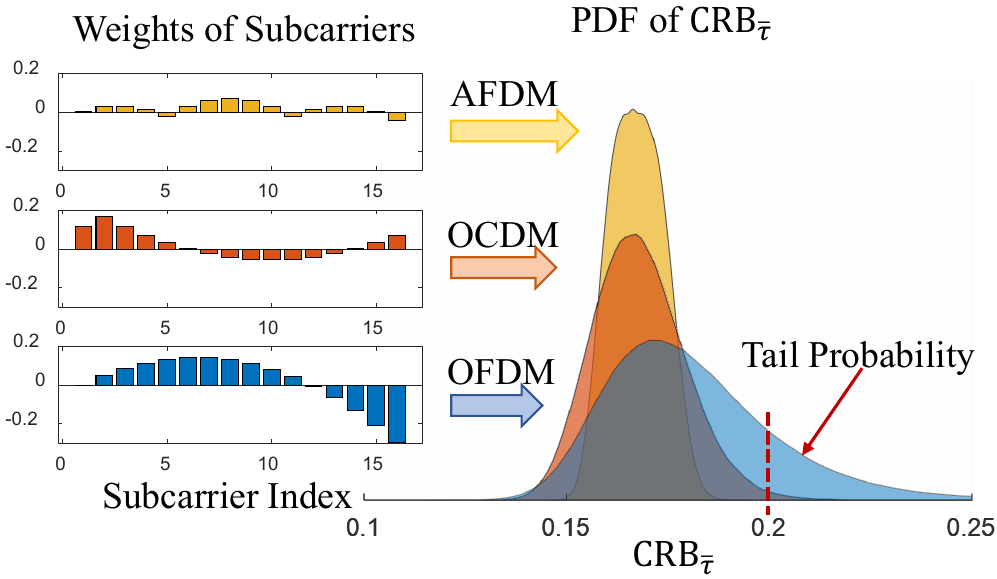}	
\vspace*{-7mm}
\caption{Comparison the PDFs of $\text{CRB}_{\bar{\tau}}$ for AFDM, OCDM and OFDM.}
\label{fig4}
\end{figure} 

Based on the above results, we investigate the robustness of the sensing CRBs for different waveforms to dynamic subcarrier power allocation strategies during communications.
We focus on $\text{CRB}_{\bar{\tau}}$. A similar conclusion can be drawn for $\text{CRB}_{\bar{\nu}}$ by using the same method, which therefore is omitted.
To characterize the sensitivity of $\text{CRB}_{\bar{\tau}}$ to the power variation of each subcarrier, we define the sensing weight of the $m$-th subcarrier as the partial derivative of $\text{CRB}_{\bar{\tau}}$ with respect to $P_m$ under equal power allocation, which is given by
\begin{align}\label{eq-CRB-va} 
  \delta_m= &\frac{\partial \text{CRB}_{\bar{\tau}}}{\partial P_m} \Big |_{(P_0,P_1,\cdots,P_{N_c-1})=\big(\frac{P_t}{N_c},\frac{P_t}{N_c},\cdots,\frac{P_t}{N_c}\big)}.  
\end{align}
When the sensing weights of subcarriers are similar and close to zero, the importance of each subcarrier to target sensing becomes approximately equal. In this case, the impact of the subcarrier power allocation variation on the sensing performance can be considered negligible. 
Conversely, if the weights differ significantly among subcarriers, their importance to target sensing varies greatly, making the sensing performance more sensitive to the subcarrier power allocation strategy.
The left section of Fig.~\ref{fig4} illustrates the subcarrier weights for AFDM ($c_1=\frac{2 \nu_{\rm{m}}+1}{2N_c}$), OCDM ($c_1=\frac{1}{2N_c}$) and OFDM ($c_1=0$), where the number of subcarriers is set to $N_c=16$. 
It can be observed that AFDM achieves the most balanced subcarrier weights with values closest to $0$, while OFDM exhibits the largest fluctuations in weight factors.
This means that $\text{CRB}_{\bar{\tau}}$ of AFDM maintains a relatively stable value even under significant variations in power allocation. 
Therefore, AFDM demonstrates more robust sensing performance in the scenarios with flexible power allocation, such as multi-user scenarios with dynamic demands.

To further explain the above conclusion, the right section of Fig.~\ref{fig4} presents the probability density function (PDF) of $\text{CRB}_{\bar{\tau}}$. 
The PDF is obtained by calculating and statistically analyzing the values of $\text{CRB}_{\bar{\tau}}$ under various random subcarrier power allocation strategies.
It can be seen that although the mean values of $\text{CRB}_{\bar{\tau}}$ for the three waveforms are similar, their variances differ significantly. OFDM exhibits the highest tail probability due to its largest variance, which increases the likelihood of degraded sensing performance. 
In contrast, AFDM has a negligible tail on the PDF curve, which makes it more likely to maintain the desired level of sensing performance. 

\subsection{Ambiguity Function Analysis}\label{S3.2}
 
The CRB only provides a theoretical lower bound of sensing deviation. To explore the inherent characteristics of AFDM waveform in terms of both sensing and communication capabilities, we derive the ambiguity function of the pilot-assisted AFDM waveform and analyze its statistical properties. 

Ambiguity function is defined as the two-dimensional correlation function in the delay-Doppler domain, which is essential for evaluating the sensing performance of various waveforms. Ambiguity functions can be classified into continuous and discrete ones, where the latter is considered as a sampled version of the former at integer points. For simplicity, this paper mainly investigates the discrete ambiguity function $\chi(\tau,\nu)$ of the $N$-length sequence {\color{red}$s(n)$}, calculated as \cite{af}
\begin{align}\label{caf} 
  \chi(\tau,\nu) =& \sum_{n=0}^{N-1} s^{*}[n]s[n-\tau]e^{\textsf{j}2\pi \nu n/N},
\end{align}
where $\tau$ and $\nu$ denote the integer delay and Doppler indices, respectively. Since the CPP part only introduces a constant multiplicative factor in the derivation of ambiguity function,  we omit it for brevity. Comparing (\ref{caf}) with (\ref{rdm}), 
the RDF $E(\tau,\nu)$ is equivalent to $\chi(\tau,\nu)$ shifted by $(\bar{\tau},\bar{\nu})$ and multiplied by a constant without consideration of noise.
As we focus on the region $\Omega_{\rm{R}}\!=\![0,\, \tau_{\rm{m}}]\!\times\![\!-\nu_{\rm{m}},\,\nu_{\rm{m}}]$ in the RDF, we only consider the region in $\chi(\tau,\nu)$ that can be shifted by $(\bar{\tau},\bar{\nu})$ into the $\Omega_{\rm{R}}$, which can be expressed as $\Omega_{\rm{A}}\!=\![-\tau_{\rm{m}},\, \tau_{\rm{m}}]\!\times\![\!-2\nu_{\rm{m}},\, 2\nu_{\rm{m}}]$.

By taking the expression of a single AFDM symbol into (\ref{caf}), the ambiguity function can be derived as
\begin{align}\label{single_a} 
	& \chi(\tau,\nu)\! =\! \frac{1}{N_c}\! \sum_{n=0}^{N_c\!-\!1}\! \bigg(\! \sum_{m_1=0}^{N_c\!-\!1} x^*[m_1] e^{-\textsf{j} 2\pi (c_1 n^2 + \frac{m_1 n}{N_c} + c_2 {m_1}^2 )}\! \bigg) \nonumber \\
  & \hspace*{3mm}\times \! \bigg(\!\sum_{m_2=0}^{N_c-1}\! x[m_2] e^{\textsf{j}2\pi(c_1 (n\!-\!\tau )^2 + \frac{m_2(n\! -\!\tau)}{N_c} + c_2 {m_2}^2)}\!\bigg) e^{\textsf{j}2\pi \nu n/N_c} \nonumber \\ 
  & = C_2\! \sum_{m_1=0}^{N_c-1}\! \sum_{m_2=0}^{N_c-1}\!\! x^{*}[m_1] x[m_2] S^{(\tau,\nu)}(m_1,m_2) e^{-\frac{\textsf{j}2\pi m_2 \tau}{N_c}}\! ,\!
\end{align}
where $C_2\! =\! e^{\textsf{j}2\pi c_1 {\tau}^2}/N_c$, and $S^{(\tau,\nu)}(m_1,m_2)$ is the interference coefficient between the $m_1$-th and $m_2$-th subcarriers given by
\begin{align}\label{eq-IintCoe} 
  S^{(\tau,\nu)}(m_1,m_2) \!\triangleq\!\! \sum_{n=0}^{N_c-1} \!e^{\textsf{j}2\pi\left(\!c_2 ({m_2}^2-{m_1}^2)-(2c_1\tau+\frac{m_1}{N_c}-\frac{m_2}{N_c}-\frac{\nu}{N_c}\!)n\!\right)}.
\end{align}
For simplicity, we assume that $2c_1N_c$ is always an integer \cite{AFDM}. Then $S^{(\tau,\nu)}(m_1,m_2)$ can be further calculated as
\begin{equation}\label{s} 
  S^{(\tau,\nu)}(\!m_1,m_2\!)\! \!=\!\!
   \begin{cases}
   \!N_ce^{\textsf{j}2\pi c_2 ({m_2}^2\!-\!{m_1}^2)\!},~ \text{if }  
   \langle  m_2\!\!-\!\!m_1\!\rangle_{N_c}\!\!=\!\text{loc}^{(\tau,\nu)}, \\
   \!0, ~\text{otherwise},
   \end{cases}
\end{equation}
where $\text{loc}^{(\tau,\nu)}\! =\! 2c_1\tau N_c\! -\! \nu$ denotes the subcarrier offset incurred by $\tau$ and $\nu$. 

Based on the aforementioned results, several theorems are provided to reveal the statistical properties of the ambiguity function of the considered AFDM waveform.

\begin{theorem}\label{T1} 
The ambiguity function of the pilot-assisted AFDM waveform can be expressed as
\begin{align}\label{e30} 
  & \chi(\tau,\nu) = \chi_{\rm{p}}(\tau,\nu) + \chi_{\rm{d}}(\tau,\nu) \nonumber \\
  & \hspace*{2mm} + C_2 \bigg(\! \sum_{m_1=0}^{N_c\!-\!1}\! \sum_{m_2=0}^{N_c\!-\!1}\! x_d^{*}[m_1] x_p[m_2] e^{-\textsf{j}2\pi\frac{m_2\tau}{N_c}} S^{(\tau,\nu)}(m_1,m_2)  \nonumber \\
  & \hspace*{2mm} + \!\sum_{m_1=0}^{N_c\!- \!1}\! \sum_{m_2=0}^{N_c\!-\!1}\! x_p^{*}[m_1] x_d[m_2] e^{-\textsf{j}2\pi\frac{m_2\tau}{N_c}} S^{(\tau,\nu)}(m_1,m_2)\!\! \bigg)\! ,\!
\end{align}
where $\chi_{\rm{p}}(\tau,\nu)$ and $\chi_{\rm{d}}(\tau,\nu)$ are the ambiguity functions of pilot and data signals, respectively. 
\end{theorem}

\begin{proof}
See Appendix~\ref{app1}.
\end{proof}

For brevity, we denote the third and forth terms of (\ref{e30}) as $\chi_{\rm{dp}}(\tau,\nu)$ and $\chi_{\rm{pd}}(\tau,\nu)$, and thus $\chi(\tau,\nu)$ can be expressed as 
\begin{align}\label{e30A} 
  \chi(\tau,\nu) =& \chi_{\rm{p}}(\tau,\nu) \!+\! \chi_{\rm{d}}(\tau,\nu) \!+\!\chi_{\rm{dp}}(\tau,\nu)\! +\! \chi_{\rm{pd}}(\tau,\nu).
\end{align}

To derive the expectation and variance of $\chi(\tau,\nu)$, we first present the following assumptions.
\begin{enumerate}
\item The constellation diagrams considered have at least three distinct phase values, and all the constellation
points are symmetrically distributed with respect to the center, leading to
\begin{align} 
  \mathbb{E}\{x_{\rm{d}}[m]\} =& 0, \, m = 0,1,\cdots, N_c-1, \label{eq-AM1-1} \\
  \mathbb{E}\{x_{\rm{d}}^2[m]\} =& 0, \, m = 0,1,\cdots, N_c-1. \label{eq-AM1-2}
\end{align}
Except for binary phase shift keying (BPSK), most common modulation formats, such as quadrature phase shift keying (QPSK) and 16-quadrature amplitude modulation (16-QAM), satisfy   assumption 1.
\item The symbols modulated on different subcarriers are independent and identically distributed, resulting in
\begin{align}\label{eq-AM2} 
  \mathbb{E}\{x_{\rm{d}}[m_1]x_{\rm{d}}^*[m_2]\} = 0, ~ m_1\neq m_2.
\end{align}
\end{enumerate}
Based on these two assumptions, the expectation of $\chi(\tau,\nu)$ is given by
\begin{align}\label{expectation} 
  \mathbb{E}\{\chi(\tau,\nu)\} = & \chi_{\rm{p}}(\tau,\nu) + \mathbb{E}\{\chi_{\rm{d}}(\tau,\nu)\} \nonumber \\
  = & \begin{cases}
        P_{\rm{t}}, ~ \text{if}~(\tau,\nu)\!\!=\!\!(0,0),\\
        \chi_{\rm{p}}(\tau,\nu) , ~\text{otherwise} ,
      \end{cases}
\end{align}
while the variance of $\chi(\tau,\nu)$ can be calculated as
\begin{align}\label{sigma} 
  \sigma^2_{\chi(\tau,\nu)} \!=\!& \begin{cases}
    2  \sigma_d^2  \sigma_p^2 
    \!+  \! (\mathbb{E}\{|x[0]|^4\}\!-\!\sigma_d^4)   N_c,~\text{if}(\tau,\nu)\!\!=\!\!(0,0), \\
    2\sigma_d^2    \sigma_p^2
    \!+ \!  \sigma_d^4 N_c ,  ~\text{otherwise}.
  \end{cases}
\end{align}
The derivation of $\mathbb{E}\{\chi(\tau,\nu)\}$ is straightforward, and the derivation of $\sigma^2_{\chi(\tau,\nu)}$ is given in Appendix~\ref{ApC}.

\begin{theorem}\label{T2} 
Given the number of subcarriers $N_c$, pilot power $\sigma^2_{\rm{p}}$ and total data symbol power $P_{\rm{d}}\! =\! N_c\sigma^2_{\rm{d}}$, the ambiguity function variance $\sigma^2_{\chi(\tau,\nu)}$ of constant-modulus modulation schemes (e.g., QPSK) is always no larger than that of non-constant modulus modulation counterparts for arbitrary $(\tau,\nu)$.
\end{theorem}

\begin{proof}
See Appendix~\ref{ApD}.
\end{proof}

\begin{theorem}\label{T3}
Given the pilot power $\sigma^2_{\rm{p}}$ and total data power $P_{\rm{d}}$, when employing constant modulus modulation schemes, $\sigma^2_{\chi(\tau,\nu)}$ converges to zero as the number of subcarriers approaches infinity: 
\begin{align}\label{eq-T3} 
  \lim_{N_c \to \infty}\sigma^2_{\chi(\tau,\nu)} = 0 .
\end{align}
\end{theorem} 

\begin{proof}
See Appendix~\ref{ApE}.
\end{proof} 
 
\begin{remark}\label{Rk1}
The inherent data randomness results in fluctuated mainlobe and sidelobes for the ambiguity function of AFDM-enabled ISAC signals, leading to possible missed detection or false alarms of targets. 
To reduce $\sigma^2_{\chi(\tau,\nu)}$ for sensing performance enhancement, we can employ constant modulus modulation schemes and increase the number of subcarriers according to Theorems~\ref{T2} and \ref{T3}. 
\end{remark}

\section{Pilot Design for Sensing and Communication Performance Enhancement}\label{S4}

In addition to the variance, the expectation of the ambiguity function is also crucial for the target sensing performance. To minimize the false alarm probability, an ideal ambiguity function should have a prominent peak at $(0,0)$ and zero value at other grid points within $\Omega_{\rm{A}}$, i.e.,
\begin{align}\label{expectation1} 
  \chi(\tau,\nu) = \begin{cases}
    P_{\rm{t}}, ~\text{if}~(\tau,\nu) = (0,0),\\
    0 , ~\text{if}~(\tau,\nu)\in \Omega_{A} \neq (0,0).
  \end{cases}
\end{align}
As can be seen from (\ref{expectation}), the expectation of the ambiguity function for the pilot-assisted AFDM signal is equal to the total power of the signal at $(0,0)$, while the values at other grid points are determined by the ambiguity function of the pilot signal. This phenomenon necessitates tailored design for the pilot sequence of the AFDM-enabled ISAC waveform.
Before elaborating on the design, the following theorem validates that the AFDM pilots with an ideal ambiguity function can guarantee desirable communication performance level.

\begin{theorem} \label{th4}
When $2c_1N_c$ is an integer, the pilot sequence with an ideal ambiguity function can achieve the optimal channel estimation performance by minimizing the mean squared error $\mathbb{E}\{\|\hat{\bm{\alpha}} - \bm{\alpha}\|^2\}$.
\end{theorem}

\begin{proof}
See Appendix~\ref{ApF}.
\end{proof}

To ensure an ideal ambiguity function for the pilot sequence, existing studies insert sufficient zero guard samples between adjacent non-zero pilots in the DAFT domain. 
Specifically, according to (\ref{s}), to ensure the interference coefficient between two adjacent non-zero pilots is always zero within $\Omega_{\rm{A}}$, the neighboring pilot spacing $Q$ should satisfy $Q\! >\! \max_{\tau,\nu} \text{loc}^{(\tau,\nu)}$.
Since $\Omega_{\rm{A}}\! = \![-\tau_{\rm{m}},\,\tau_{\rm{m}}]\!\times\![-2\nu_{\rm{m}},\, 2\nu_{\rm{m}}]$, the subcarrier offset $\text{loc}^{(\tau,\nu)}\! \in\! [-2c_1\tau_{\rm{m}} N_c-2\nu_{\rm{m}},\, 2c_1\tau_{\rm{m}} N_c+2\nu_{\rm{m}}]$, and the ambiguity function of the pilot sequence becomes ideal when the adjacent non-zero pilot spacing $Q$ satisfies:
\begin{equation}\label{q} 
  Q \geq 2c_1\tau_{\rm{m}} N_c+2\nu_{\rm{m}} + 1 .
\end{equation}
For example, in {\cite{super_pilot}}, $Q$ is set to $2c_1\tau_{\rm{m}} N_c+2\nu_{\rm{m}} +1$ and the number of non-zero pilots is chosen to be $N_{\rm{p}}\! =\!\big\lfloor\frac{N_c}{Q} \big\rfloor$.
Inversely, when the pilot spacing $Q$ is fixed, the maximum range of sensing and channel estimation is limited by
\begin{equation}\label{q1}  
  \tau_{\rm{m}} \leq \frac{Q-2\nu_{\rm{m}}-1}{2c_1 N_c}. 
\end{equation}
According to (\ref{q1}), small $Q$ limits the capability boundary of both target sensing and channel estimation. On the other hand, the margin $\tau_{\rm{m}}$ can be expanded by enlarging Q, which however sacrifices the number of non-zero pilots within a single AFDM symbol, causing degraded channel estimation accuracy. This dilemma makes the pilot assignment for AFDM rather challenging.

To overcome this challenge, we propose a novel pilot design in Theorem~\ref{T5}, which includes the choice of the parameter $c_1$ and the pilot spacing configuration as well as the pilot sequence design.
The proposed pilot sequence features an ideal ambiguity function with the spacing $Q$ independent of $\tau_{\rm{m}}$. As a result, our design enables simultaneous channel estimation and long-range target sensing both with superior accuracy.

\begin{theorem}\label{T5}
Consider an AFDM symbol with $N_c\! =\! 2^p$ subcarriers, where $p$ is a non-zero integer. When $c_1\! =\! \frac{2^{q}}{2N_c}$ with the integer $q$ satisfying $2^{q-1}\! <\! 2\nu_{\rm{m}}+1\leq2^{q}$, $Q\! =\! 2c_1N_c{2^r}$ with $r\! =\! 0,1,\cdots ,p - q$, and $N_{\rm{p}}\! =\! N_c/Q$,
the pilot signal satisfying the following equation has an ideal ambiguity function:
\begin{align}\label{pilot2} 
  x_{\rm{p}}[m] \!=\!& \begin{cases}
     \sqrt{\frac{\sigma^2_{\rm{p}}}{N_{\rm{p}}}} z\Big[\!\frac{m}{Q}\!\Big] e^{\textsf{j}2\pi\psi[m]  } ,~ m=0,Q,\cdots,(N_{\rm{p}}\!\!-\!\!1)Q, \\
     0,  ~\text{otherwise},
   \end{cases}
\end{align}
where $\psi[m] = \left(\!\frac{m^2  {2^r}}{2QN_c} - c_2 {m}^2\!\right)$, and $z[k]$ denotes any $N_{\rm{p}}$-length CAZAC sequence\cite{Popovic_tit_18}, such as the Zadoff-Chu (ZC) sequence\cite{zc}.
\end{theorem}

\begin{proof}
See Appendix~\ref{ApG}.
\end{proof}

\begin{remark}\label{Rm2}
The advantage of the proposed pilot design lies in the fact that $Q$ is independent of the maximum delay $\tau_{\rm{m}}$. As shown in Theorem~\ref{th4}, $Q$ can be as small as $2c_1 N_c$, where $c_1$ depends solely on the maximum Doppler shift. Therefore, as $\tau_{\rm{m}}$ increases, it is unnecessary to reduce the number of non-zero pilots, thus maintaining accurate channel estimation. 
\end{remark}

In the following, we will give a specific example to illustrate how to obtain a pilot sequence based on Theorem~5. Firstly, an appropriate $N_{\rm{p}}$-length CAZAC sequence should be chosen as the base sequence. Let us take the $N_{\rm{p}}$-length ZC sequence as an example. According to Theorem~5, $N_{\rm{p}}$ is an even number, and thus the $N_{\rm{p}}$-length ZC sequence can be written as
\begin{align}\label{pilot3}
z[ n ]=e^{-\frac{\textsf{j} \pi u n^2}{N}  } ,~n=0,1,\cdots,N_{\rm{p}}\!\!-\!\!1,
\end{align}
where $u$ is an integer satisfying $\gcd(N_{\rm{p}},u)=1$. Next, by substituting (\ref{pilot3}) in (\ref{pilot2}), the pilot signal can be expressed as 
\begin{align}\label{pilot4}  
  x_{\rm{p}}[m] \!=\! & \begin{cases}
     \sqrt{\frac{\sigma^2_{\rm{p}}}{N_{\rm{p}}}} e^{-\frac{\textsf{j} \pi u m^2}{NQ^2}  } e^{\textsf{j}2\pi\psi[m]  } ,~ m=0,Q,\cdots,(N_{\rm{p}}\!\!-\!\!1)Q, \\
     0,  ~\text{otherwise},
   \end{cases}
\end{align}
Finally, we superpose the pilot signal on the data as (\ref{signal}) to obtain the transmit signal.

It should be noted that the proposed pilot design methodology is not equivalent to directly transforming the CAZAC sequence into the DAFT domain. CAZAC sequences ensure perfect auto-correlation property only when the Doppler shift is zero, and their autocorrelation performance deteriorates significantly under large Doppler shifts. Consequently, the sequence derived from directly transforming a CAZAC sequence cannot guarantee an ideal ambiguity function. In contrast, our proposed pilot is designed by exploiting the unique characteristics of AFDM.
Due to its inherent multipath separation capability, interference between non-zero pilot symbols arises only from specific combinations of delay and Doppler shifts. We identify these critical delay-Doppler pairs and introduce appropriate phase shifts to the CAZAC sequences, enabling them to maintain ideal auto-correlation properties even under substantial Doppler shifts. 

\begin{figure}[!t]
\vspace*{-4mm}
\center
\includegraphics[width=.85\linewidth, keepaspectratio]{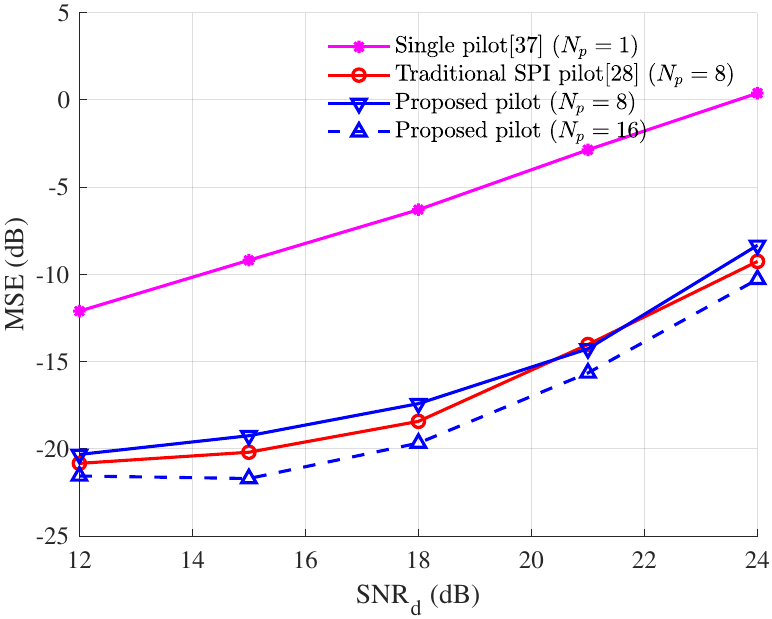}	
\vspace*{-3mm}
\caption{Comparison of channel estimation performance between the proposed pilot design and existing designs under the scenario of $\tau_{\rm{m}}\! =\! 2$.} 
\label{sim2} 
\vspace*{-1mm}
\end{figure} 

\section{Numerical Results}\label{S5}

\begin{table}[!b]
\caption{Simulations Parameters}
\centering
\begin{tabular}{|c|c|c|}
\hline
Symbol & Parameter & Value\\
\hline
$f_c$ & Carrier frequency & $28$ GHz\\
\hline
$N_c$ & Number of subcarriers & $128$ \\
\hline
$N_{cp}$ & Number of CPP & $32$ \\
\hline
$\Delta f$ & Subcarrier spacing & $100$ KHz \\
\hline
$\nu_{\rm{m}}$ & Maximum normalized Doppler shift & $2$ \\
\hline
$\tau_{\rm{m}}$ & Maximum normalized delay & $2$ or $15$ \\
\hline
$\sigma^2_{\rm{p}}$ & Power of pilot & $20$ dB\\
\hline
$\sigma^2_{\rm{cn}}$ & Power of noise & $30$ dBm\\
\hline
\end{tabular}
\end{table}

Consider an AFDM system with $N_c\! =\! 128$ subcarriers, where each subcarrier carries a symbol modulated using QPSK. The key simulation parameters are summarized as Table~~\uppercase\expandafter{\romannumeral1}. We consider both the large-range scenario ($\tau_{\rm{m}}\! =\! 15$) and the small-range scenario ($\tau_{\rm{m}}\! =\! 2$). The maximum normalized Doppler frequency shift is set to $\nu_{\rm{m}}\! =\! 2$. 
The communication channel consists of $L\! =\! 3$ paths, with the normalized delay and Doppler shift of each path randomly selected from $[0,\, \tau_{\rm{m}}]$ and $[-\nu_{\rm{m}},\, \nu_{\rm{m}}]$, respectively. 
The channel coefficient for each path follows a complex Gaussian distribution $\mathcal{CN}\! \sim\!  (0,1/L)$. 
For the sensing purpose, we consider a single-target model, where its normalized delay and Doppler shift are uniformly distributed within the ranges of $[0,\, \tau_{\rm{m}}]$ and $[-\nu_{\rm{m}},\, \nu_{\rm{m}}]$, respectively. Numerical results are provided to validate the feasibility and superiority of the proposed pilot design in terms of both communication and sensing performance.

\subsection{Communication Performance}\label{S5.1}

Figs.~\ref{sim2} and \ref{sim4} depict the channel estimation performances as the functions of the data signal-to-noise ratio $\text{SNR}_{\rm{d}}\! =\! \sigma_{\rm{d}}^2/{\sigma_{\rm{cn}}^2}$ for the proposed pilot design and two existing designs under the scenarios of $\tau_{\rm{m}}\! =\! 2$ and $\tau_{\rm{m}}\! =\! 15$, respectively. 
The total power of pilot and noise signals are set as $\sigma_{\rm{p}}^2=20$\,dB and $N_c\sigma_{\rm{cn}}^2=21$\,dBm, respectively. The mean square error (MSE) is utilized to evaluate the channel estimation performance. It can be obtained by calculating the Frobenius norm of the difference between the channel matrix and its estimated value, which can be expressed as
\setcounter{equation}{56}
\renewcommand\theequation{\arabic{equation}}
\begin{align}
    \text{MSE}=\sum_{i=0}^{N_{\text{iter}}}\frac{|| \mathbf{H}^{(i)}_{\rm{c}} - \hat{\mathbf{H}}^{(i)}_{\rm{c}}||}{N_{\text{iter}}}.
\end{align}
Here $\mathbf{H}^{(i)}_{\rm{c}}$ and $\hat{\mathbf{H}}^{(i)}_{\rm{c}}$ denote the channel matrix and its estimated value for the $i$-th simulation run, and the results are averaged over $N_{\text{iter}}=10000$ simulation runs.
In the single pilot design \cite{arxiv}, there is only one non-zero pilot in the DAFT domain, which is put in the first subcarrier. The traditional SPI pilot denotes the superimposed pilot in \cite{super_pilot} with $N_{\rm{p}}\! =\! 8$ non-zero pilots.
For the proposed pilot design, as shown in (\ref{pilot2}), we consider $N_{\rm{p}}\! =\! 8$ and $N_{\rm{p}}\! =\! 16$. For a fair comparison, all the schemes adopt the same AFDM parameters with $c_1\! =\! 4/N_c$ and $c_2\! = \!\pi\!-\!3$.

\begin{figure}[!t]
\vspace*{-4mm}
\center
\includegraphics[width=.85\linewidth, keepaspectratio]{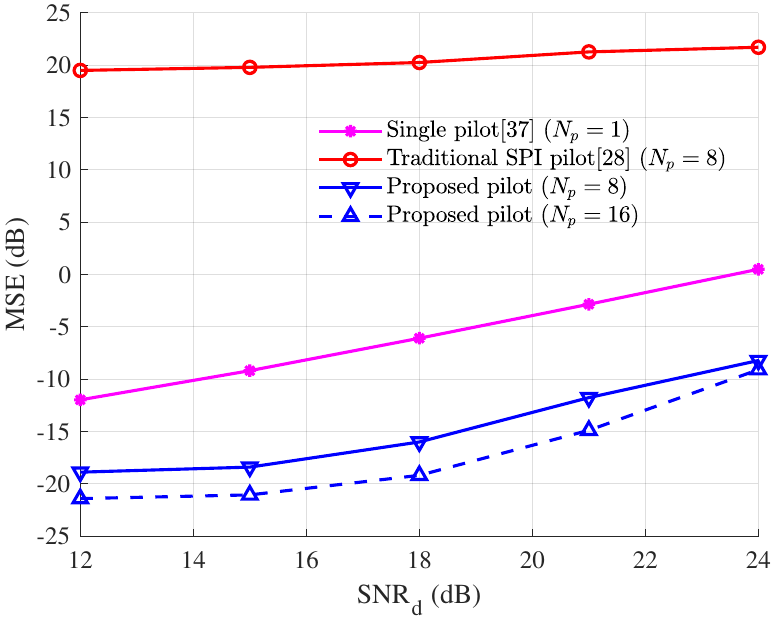}	
\vspace*{-3mm}
\caption{Comparison of channel estimation performance between the proposed pilot design and existing designs under the scenario of $\tau_{\rm{m}}\! =\! 15$.}
\label{sim4} 
\vspace*{-4mm}
\end{figure} 

It can be seen from Fig.~\ref{sim2} that as the data power increases, the equivalent noise power in channel estimation also increases, resulting in higher MSE for all the schemes. For the small-range scenario of $\tau_{\rm{m}}\! =\! 2$, both the proposed pilot design and the traditional SPI pilot with $N_{\rm{p}}\! =\! 8$ non-zero pilots achieve more accurate channel estimation results, with a $7$\,dB improvement in MSE compared to the single pilot scheme. This is because, when $\tau_{\rm{m}}$ is relatively small, both the proposed pilot and traditional SPI pilot designs can realize the optimal channel estimation with no inter-pilot interference. Moreover, since the number of non-zero pilots in these designs is greater than that in the single pilot scheme, the UE can utilize more subcarriers for channel estimation, thereby achieving higher channel estimation accuracy.
As expected, the proposed pilot design with $N_{\rm{p}}\! =\! 16$ outperforms the same design with $N_{\rm{p}}\! =\! 8$.
On the other hand, for the large-range scenario of $\tau_{\rm{m}}\! =\! 15$, the traditional SPI pilot with 
$N_{\rm{p}}\! =\! 8$ no longer satisfies the requirements of (\ref{q}). The received responses of two non-zero pilot symbols may overlap in the DAFT domain after passing through the channel, which leads to severe pilot interference and prevents the channel estimation from achieving the minimum mean square error performance.
In contrast, the proposed pilot design maintains an ideal ambiguity function and ensures the optimal channel estimation with the similar MSEs in the both scenarios, as can be observed from Figs.~\ref{sim2} and \ref{sim4}. 

\begin{figure}[!t]
\center
\includegraphics[width=.85\linewidth, keepaspectratio]{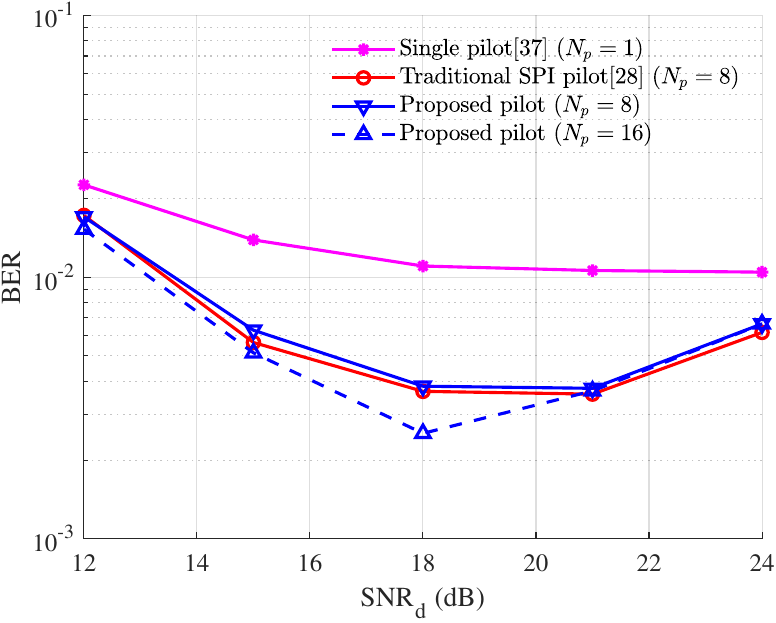}	
\vspace*{-3mm}
\caption{BER comparison between the proposed pilot design and existing designs under the scenario of $\tau_{\rm{m}}\! =\! 2$.}
\label{sim1} 
\vspace*{-4mm}
\end{figure} 

Figs.~\ref{sim1} and \ref{sim3} investigate the BER performance of the proposed pilot design and two existing designs under $\tau_{\rm{m}}\! =\! 2$ and $\tau_{\rm{m}}\! =\! 15$, respectively. 
Since the channel estimation performance of the single pilot design is considerably inferior to those of the proposed pilot and traditional SPI pilot designs, its BER is much larger than those of the other two designs.
As for our pilot and traditional SPI pilot designs, generally, for $\text{SNR}_{\rm{d}}\! <\! 18$\,dB, the BER reduces as $\text{SNR}_{\rm{d}}$ increases but for $\text{SNR}_{\rm{d}}\! >\! 21$\,dB, the BER increases as $\text{SNR}_{\rm{d}}$ increases. This is because for $\text{SNR}_{\rm{d}}\! <\! 18$\,dB, the channel estimation is relatively accurate, and increasing the data power improves the data demodulation performance. 
On the other hand, for $\text{SNR}_{\rm{d}}\! >\! 21$\,dB, as $\text{SNR}_{\rm{d}}$ increases, the equivalent noise power for the channel estimation increases, leading to a decline in channel estimation accuracy and a corresponding rise in the BER.
By comparing Figs.~\ref{sim1} and \ref{sim3}, it can be seen that as the maximum delay spread increases, the BER of the traditional SPI pilot increases significantly due to the degradation of its channel estimation performance. The proposed pilot design by contrast does not suffer from this performance degradation.

\begin{figure}[!h]
\center
\includegraphics[width=.85\linewidth, keepaspectratio]{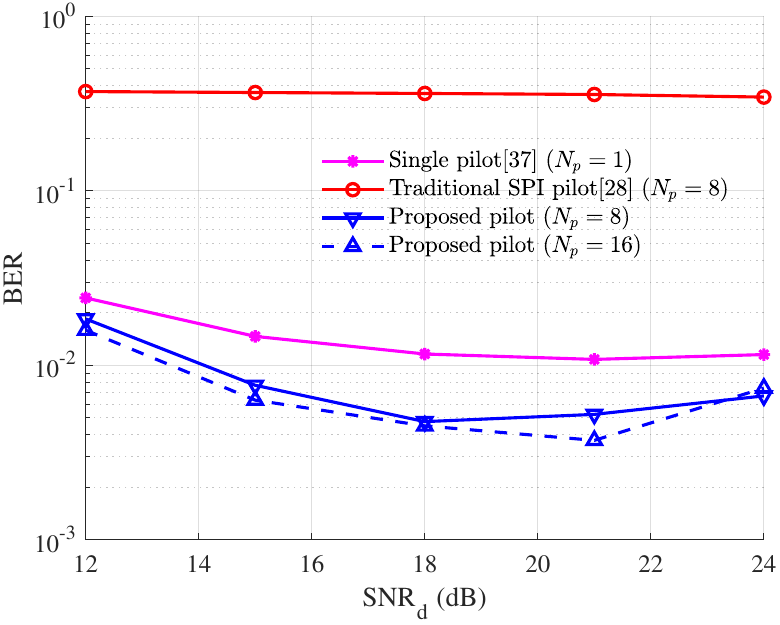}	
\vspace*{-3mm}
\caption{BER comparison between the proposed pilot design and existing designs under $\tau_{\rm{m}}=15$.}
\label{sim3} 
\end{figure} 

\subsection{Sensing Performance}\label{S5.2}

The ROC curves of the proposed pilot and traditional SPI pilot designs are illustrated in Fig.~\ref{sim5} under the large-range scenario of $\tau_{\rm{m}}\! =\! 15$, where the receive SNR at the ISAC receiver is set to $0$\,dB or $-10$\,dB. Both the proposed pilot and traditional SPI pilot designs have $N_{\rm{p}}\! =\! 8$ non-zero pilots.
To obtain the ROC curve, we vary the threshold $\gamma$ in (\ref{eqHtest}) and record the corresponding probabilities of false alarm and missed detection, as shown in Fig.~\ref{sim5}.
We define a missed detection event as either failing to detect the peak corresponding to the target or when the difference between the estimated normalized delay (or normalized Doppler shift) and the true value exceeds $1$.
Since the traditional SPI pilot cannot ensure an ideal ambiguity function under $\tau_{\rm{m}}\! =\! 15$, its ROC curve is much inferior to that of the proposed pilot design, i.e., the probability of missed detection of the traditional SPI pilot is much higher than that of the proposed pilot at the same false alarm probability.
Due to its ideal ambiguity function, the performance of the proposed pilot design is not limited by high sidelobes but is primarily determined by the receive SNR. As a result, when the SNR increases, the proposed scheme exhibits a significant improvement in sensing performance.

\begin{figure}[!t]
\center
\includegraphics[width=.85\linewidth, keepaspectratio]{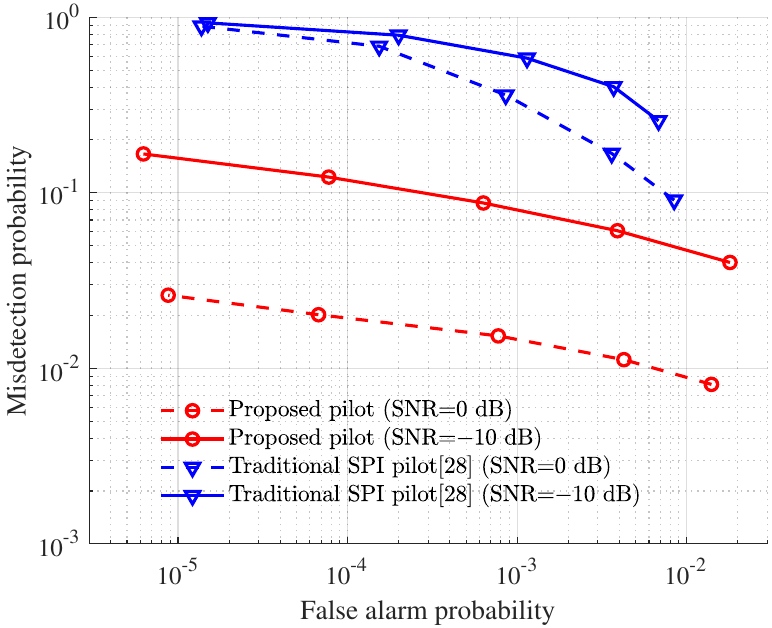}	
\vspace*{-3mm}
\caption{ROC curve comparison between the proposed pilot and traditional SPI pilot designs under $\tau_{\rm{m}}\! =\! 15$ and given two different receive SNRs.}
\label{sim5} 
\vspace*{-4mm}
\end{figure} 

Figs.~\ref{fig5}\,(a) and \ref{fig5}\,(b) compare the simulated distance root MSE ($\text{RMES}_{R}$) and velocity root MSE ($\text{RMES}_{V}$) with $\sqrt{\text{CRB}_{R}}$ and $\sqrt{\text{CRB}_{V}}$, respectively, where the carrier frequency is set to $f_c\! =\! 28$\,GHz and the sensing range is set to $\tau_{\rm{m}}\! =\! 15$. 
Specifically, $\sqrt{\text{CRB}_{R}}$ and $\sqrt{\text{CRB}_{V}}$ are calculated by averaging $\text{CRB}_{R}$ and $\text{CRB}_{V}$ for any possible $\bar{\tau}$, respectively, and deriving the square root of the results.
It can be seen from Fig.~\ref{fig5}\,(a) that as the sampling period decreases, the distance estimation error reduces, which agrees with (\ref{crbr}). 
From Fig.~\ref{fig5}\,(b), it can be seen that as the subcarrier spacing decreases, the velocity estimation error reduces, which agrees with (\ref{crbv}). 
Moreover, it can be observed that both $\text{RMES}_{R}$ and $\text{RMES}_{V}$ are close to the square roots of derived CRBs, demonstrating that the proposed pilot design enables high-precision sensing.


\begin{figure}[!t]
\vspace*{-1mm}
\center
 \subfigure[Comparison between $\rm{RMES}_{R}$ and  $\sqrt{\rm{CRB}_{R}}$.] {  
\includegraphics[width=0.85\linewidth, keepaspectratio]{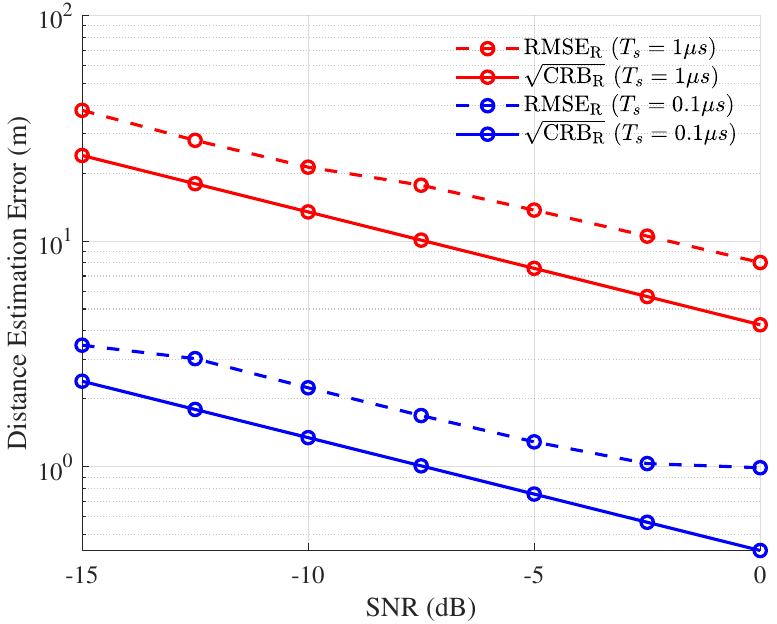}  }
 \hspace{-1mm}
 
 \subfigure[Comparison between $\rm{RMES}_{V}$ and $\sqrt{\rm{CRB}_{V}}$.] {   
 \includegraphics[width=0.85\linewidth, keepaspectratio]{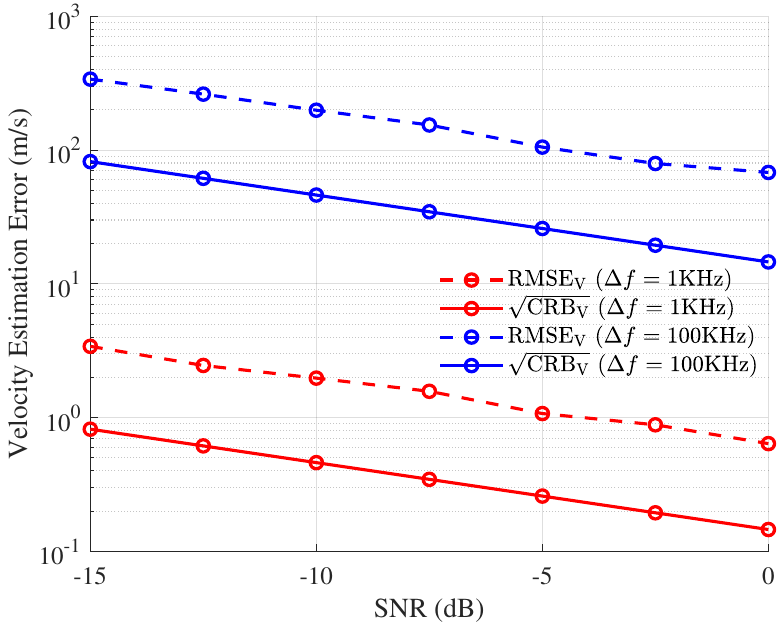}  
 }
\vspace*{-1mm}
\caption{Comparison of the simulated RMSE and the square root of CRB.}
\label{fig5}  
\vspace{-4mm}
\end{figure}  

\section{Conclusions}\label{S6}

In this paper, we have presented a theoretical framework for AFDM-enabled ISAC systems and introduced a novel pilot design for pilot-assisted AFDM waveforms to enhance both communication and sensing performance. Specifically, we derived closed-form CRB expressions for target sensing in AFDM, offering helpful guidelines for system design and insights into the advantages of AFDM over existing alternatives. 
Furthermore, we have formulated the ambiguity function of the pilot-assisted AFDM waveform for the first time and analyzed its statistical properties. To overcome the sensing range limitations imposed by pilot spacing, we have also proposed a novel superimposed pilot design, which has been analytically proven to achieve an ideal ambiguity function and optimal channel estimation simultaneously. Numerical results have been provided to validate the theoretical derivations and demonstrate the superiority of the proposed pilot design.
 
\appendix

\subsection{Derivation of FIM}\label{app0}

First we provide a detailed derivation of $\frac{\partial{\bar{s}[n]}}{\partial \bar{\tau}}$:
\begin{align} \label{g0} 
  \frac{\partial{\bar{s}[n]}}{\partial \bar{\tau}} =& \beta \frac{\partial \bar{s}(n T_s - \bar{\tau} T_s)}{\partial \bar{\tau}} e^{\textsf{j} 2\pi \bar{\nu} n/N_c}  \nonumber \\
  =& \frac{\beta e^{\textsf{j} 2\pi \bar{\nu} n/N_c}}{\sqrt{N_{c}}} \sum_{m=0}^{N_{c}-1} \bigg(\textsf{j} 2\pi \frac{\partial g_m(n T_s - \bar{\tau}T_s)}{\partial \bar{\tau}}\bigg) x[m] \nonumber \\
  & \times e^{\textsf{j} 2 \pi\left( g_m(n T_s- \bar{\tau}T_s) + c_{2} m^{2}\right)} . 
\end{align}
We note
\begin{align}\label{g1} 
  \frac{\partial g_m(nT_s\! - \!\bar{\tau}T_s)}{\partial \bar{\tau}} = & - T_s \frac{\partial g_m(t)}{\partial t}|_{t = nT_s\!- \!\bar{\tau}T_s} ,
\end{align}
with
\begin{align}\label{g2} 
  \frac{\partial g_m(t)}{\partial t} =& \frac{2 c_{1} t}{T_{s}^{2}}\! +\! \frac{m}{N_{c}T_{s}}\! -\! q_m\Big(\frac{t}{T_s}\Big)\frac{1}{T_s}\! -\! \frac{\partial q_m\big(\frac{t}{T_s}\big)}{\partial t}\frac{t}{T_s}.
\end{align}
Since $q_m\left(\frac{t}{T_s}\right)\! =\! \big\lfloor 2c_1\frac{t}{T_s} + \frac{m}{N_c} \big\rfloor$, $\frac{\partial q_m\big(\frac{t}{T_s}\big)}{\partial t}$ takes the value of $0$ in most cases, and it can be dropped from (\ref{g2}) for approximation. Therefore, we have
\begin{align}\label{g3} 
  \frac{\partial g_m(nT_s\! - \!\bar{\tau}T_s)}{\partial \bar{\tau}} \approx& -\! 2 c_{1} (n \! - \!\bar{\tau})\! -\! \frac{m}{N_{c}}\! +\! \Big\lfloor 2 c_{1} (n \! - \!\bar{\tau})\! +\! \frac{m}{N_c} \Big\rfloor \nonumber \\
  \triangleq& - \mathcal{F}_m(n - \bar{\tau}).
\end{align}
Substituting (\ref{g3}) into (\ref{g0}) yields (\ref{e23}).

Based on the above calculation results, we can compute $\mathbf{F}_{\mathbf{I}}$. Take $F_{I}[1,1]$ as an example, which can be derived as
\begin{align}\label{i11} 
  F_I[1,1] =& \frac{2}{\sigma_{\rm{s}}^2} \Re \left\{ \mathbb{E} \bigg\{ \sum_{n=0}^{N_c-1}\Big|s(n T_s - \bar{\tau}T_s) e^{\textsf{j}2\pi \bar{\nu}n/N_c}\Big|^2 \bigg\} \right\} \nonumber \\
  =& \frac{2}{\sigma_{\rm{s}}^2N_c^2} \Re \Bigg\{ \sum_{m_1=0}^{N_c-1}\!\sum_{m_2=0}^{N_c-1} \mathbb{E} \Big\{x[m_1]x^*[m_2] \Big\} \nonumber \\
  & \hspace*{5mm}\times \bigg(\sum_{n=-\bar{\tau}}^{N_c-1-\bar{\tau}} e^{\textsf{j}2\pi(\eta_{m_1}[n]-\eta_{m_2}[n])}\bigg) \Bigg\}.
\end{align}
Since $\sum_{n=-\bar{\tau}}^{N_c-1-\bar{\tau}} e^{\textsf{j}2\pi(\eta_{m_1}[n]-\eta_{m_2}[n])} =0 $ for $m_1\neq m_2$, (\ref{i11}) can be further simplified as
\begin{align}\label{i12} 
  F_I[1,1] =& \frac{2}{\sigma_{\rm{s}}^2N_c^2} \mathbb{R}\Bigg\{ N_c\!\sum_{m_1=0}^{N_c-1} \mathbb{E}\big\{|x[m_1]|^2\big\} \Bigg\} = \frac{2P_t}{\sigma_{\rm{s}}^2N_c} ,
\end{align}
where $P_m\! =\! \mathbb{E}\{|x[m]|^2\}$ and $P_t\! =\! \sum_{m=0}^{N_c-1}P_m$ .
Other elements of $\mathbf{F}_{\mathbf{I}}$ can be calculated using similar approaches, which are omitted for brevity. By combining all the elements, we obtain the FIM (\ref{fim}).

\subsection{Proof of Theorem~\ref{T1}}\label{app1}

\begin{proof}
By substituting $x[m]\! =\! x_{\rm{p}}[m]+x_{\rm{d}}[m]$ into (\ref{single_a}), we have
\begin{align}\label{e53} 
  & \chi(\tau,\nu) = \nonumber \\
  & C_2 \bigg(\! \sum_{m_1=0}^{N_c-1}\!\sum_{m_2=0}^{N_c-1}\! x_{\rm p}^{*}[m_1] x_{\rm p}[m_2] e^{-\frac{\textsf{j}2\pi m_2\tau}{N_c}} S^{(\tau,\nu)}(m_1,m_2) \nonumber \\
  & \hspace*{1mm} + \sum_{m_1=0}^{N_c-1}\! \sum_{m_2=0}^{N_c-1}\! x_{\rm d}^{*}[m_1] x_{\rm d}[m_2] e^{-\frac{\textsf{j}2\pi m_2\tau}{N_c}} S^{(\tau,\nu)}(m_1,m_2) \nonumber \\
  & \hspace*{1mm} + \sum_{m_1=0}^{N_c-1}\! \sum_{m_2=0}^{N_c-1}\! x_{\rm d}^{*}[m_1] x_{\rm p}[m_2] e^{-\frac{\textsf{j}2\pi m_2\tau}{N_c}} S^{(\tau,\nu)}(m_1,m_2) \nonumber \\
  & \hspace*{1mm} + \sum_{m_1=0}^{N_c-1}\! \sum_{m_2=0}^{N_c-1}\! x_{\rm p}^{*}[m_1] x_{\rm d}[m_2] e^{-\frac{\textsf{j}2\pi m_2\tau}{N_c}} S^{(\tau,\nu)}(m_1,m_2) \!\bigg).
\end{align}
According to the definition of ambiguity function, the first and second terms of (\ref{e53}) are the ambiguity functions of the pilot ($\chi_{\rm{p}}(\tau,\nu)$) and data signal ($\chi_{\rm{d}}(\tau,\nu)$), respectively. This completes the proof. 
\end{proof}

\subsection{Derivation of $\sigma^2_{\chi(\tau,\nu)}$}\label{ApC}

The variance of ${\chi(\tau,\nu)}$ is given by
\begin{align}\label{e57} 
  \sigma^2_{\chi(\tau,\nu)} =& \mathbb{E}\{|\chi(\tau,\nu)|^2\} - |\mathbb{E}\{\chi(\tau,\nu)\}|^2 .
\end{align}
Given the assumptions $\mathbb{E}\{x_{\rm{d}}[m]\}=0$ and $\mathbb{E}\{x^2_{\rm{d}}[m_1]\} =0$, any term of $\mathbb{E}\{|\chi(\tau,\nu)|^2\}$ containing $\mathbb{E}\{x_{\rm{d}}[m]\}$ and $\mathbb{E}\{x^2_{\rm{d}}[m_1]\}$ is zero. Therefore,
\begin{align}\label{e58} 
  \mathbb{E}\{|\chi(\tau,\nu)|^2\} =& |\chi_{\rm{p}}(\tau,\nu)|^2 + \mathbb{E}\{|\chi_{\rm{d}}(\tau,\nu)|^2\} \nonumber \\
	& + \mathbb{E}\{|\chi_{\rm{dp}}(\tau,\nu)|^2\} + \mathbb{E}\{|\chi_{\rm{pd}}(\tau,\nu)|^2\} \nonumber \\
	& + \mathbb{E}\big\{ 2\Re\{\chi_{\rm{p}}(\tau,\nu)\chi^*_{\rm{d}}(\tau,\nu)\}\big\}.
\end{align}
According to (\ref{expectation}), 
\begin{align}\label{e59} 
  |\mathbb{E}\{\chi(\tau,\nu)\}|^2 =& |\chi_{\rm{p}}(\tau,\nu)|^2 + |\mathbb{E}\{\chi_{\rm{d}}(\tau,\nu)\}|^2 \nonumber \\
  & + \mathbb{E}\big\{ 2\Re\{\chi_{\rm{p}}(\tau,\nu)\chi^*_{\rm{d}}(\tau,\nu)\}\big\}.
\end{align}
By substituting (\ref{e58}) and (\ref{e59}) into (\ref{e57}), we obtain
\begin{align}\label{e60} 
  \sigma^2_{\chi(\tau,\nu)}\! =\! \sigma^2_{\chi_{\rm{d}}(\tau,\nu)}\! +\! \mathbb{E}\{|\chi_{\rm{dp}}(\tau,\nu)|^2\}\! +\! \mathbb{E}\{|\chi_{\rm{pd}}(\tau,\nu)|^2\} .
\end{align}
The first term in (\ref{e60}) can be derived as 
\begin{align}\label{eq-Term1} 
  \sigma^2_{\chi_{\rm{d}}(\tau,\nu)} = & \frac{1}{N_c^2} \bigg( \sum_{m_1=0}^{N_c-1}\sum_{m_2 \neq m
 _1}^{N_c-1} \sigma_d^4 |S^{(\tau,\nu)}(m_1,m_2)|^2 \nonumber \\
  & \hspace*{-5mm}+ \sum_{m=0}^{N_c-1}\big(\mathbb{E}\{|x[m]|^4\}-\sigma_d^4\big) |S^{(\tau,\nu)}(m,m)|^2 \bigg)\! .\!
\end{align}
The second term in (\ref{e60}) is equal to the third term, which can be expressed as
\begin{align}\label{eq-Term2} 
  \mathbb{E}\{|\chi_{\rm{dp}}(\tau,\nu)|^2\}\! = & \frac{1}{N_c^2}\!\! \sum_{m_1=0}^{N_c-1} \sum_{m_2=0}^{N_c-1}\!\! \sigma_d^2 |x_{\rm{p}}[m_2]|^2 |S^{(\tau,\nu)}(m_1,m_2)|^2 \! .
\end{align}
According to (\ref{s}), $|S^{(\tau,\nu)}(m_1,m_2)|^2$ can be calculated as
\begin{equation}\label{s1} 
  |S^{(\tau,\nu)}(m_1,m_2)|^2\! =\!
   \begin{cases}
     \! N_c^2, \, \text{if }  
     \langle m_2\!-\!m_1\rangle_{N_c}\!\! =\! \text{loc}^{(\tau,\nu)}\! ,\! \\
     \! 0, ~\text{otherwise}.
   \end{cases}
\end{equation}
Therefore, when $(\tau,\nu)=(0,0)$, we obtain
\begin{align}\label{e65} 
  \sigma^2_{\chi(\tau,\nu)} =& N_c\big(\mathbb{E}\{|x[m]|^4\}-\sigma_d^4\big) + 2 \sigma_d^2 \sum_{m_1=0}^{N_c-1} |x_{\rm{p}}[m_1]|^2 \nonumber \\
  =& N_c\big(\mathbb{E}\{|x[0]|^4\}-\sigma_d^4\big) + 2 \sigma_d^2 \sigma_{\rm{p}}^2.
\end{align}
When $(\tau,\nu)\neq (0,0)$, we have
\begin{align}\label{e66}  
  \sigma^2_{\chi(\tau,\nu)} =& \sum_{m_1=0}^{N_c-1} \sigma_d^4 + 2 \sigma_d^2 \sum_{m_1=0}^{N_c-1} \big|x_{\rm{p}}[\langle m_1 + \text{loc}^{(\tau,\nu)}\rangle_{N_c}]\big|^2 \nonumber \\
  =& N_c \sigma_d^4 + 2 \sigma_d^2 \sigma_{\rm{p}}^2.
\end{align}

\subsection{Proof of Theorem~\ref{T2}}\label{ApD}

\begin{proof}
As seen in (\ref{sigma}), when $(\tau,\nu)\! \neq\! (0,0)$, the variance is irrelevant to specific modulation method. On the other hand, when $(\tau,\nu)\! =\! (0,0)$, the variance can be reduced by minimizing the fourth moment of the constellation's amplitude, i.e., $\mathbb{E}\{|x[0]|^4\}$. 
According to the Cauchy-Schwarz inequality, the fourth moment is no less than the square of the second moment, i.e., $\mathbb{E}\{|x[0]|^4\}\! \geq\! \sigma^4_{\rm{d}}$.
Furthermore, $\mathbb{E}\{|x[0]|^4\}\! =\! \sigma^4_{\rm{d}}$ if and only if the constellation points have constant modulus.
Therefore, by employing the constant-modulus modulation, the ambiguity function variance is minimized for $(\tau,\nu)\! =\! (0,0)$. This completes the proof.
\end{proof}
 
\subsection{Proof of Theorem~\ref{T3}}\label{ApE}

\begin{proof}
By choosing $\sigma^2_{\rm{d}}=P_{\rm{d}}/N_c$ and $\mathbb{E}\{|x[0]|^4\} = \sigma_d^4$ in (\ref{sigma}), the variance can be written as
\begin{align}\label{sigma1} 
  \sigma^2_{\chi(\tau,\nu)} =& \begin{cases}
    2 \frac{P_{\rm{d}} \sigma_p^2}{N_c}, ~ \text{if}(\tau,\nu) = (0,0), \\
    2 \frac{P_{\rm{d}} \sigma_p^2}{N_c} + \frac{P^2_{\rm{d}}}{N_c}, ~ \text{otherwise}.
  \end{cases}
\end{align}
For $(\tau,\nu)=(0,0)$,
\begin{align}\label{eq-Lim1} 
  \lim_{N_c\to \infty} \sigma^2_{\chi(\tau,\nu)} = \lim_{N_c\to \infty} 2 \frac{P_{\rm{d}} \sigma_p^2}{N_c} = 0.
\end{align}
For $(\tau,\nu)\neq(0,0)$, 
\begin{align}\label{eq-Lim2} 
  \lim_{N_c\to \infty} \sigma^2_{\chi(\tau,\nu)} = \lim_{N_c\to \infty} 2\frac{P_{\rm{d}} \sigma_p^2}{N_c} +  \frac{P^2_{\rm{d}}}{N_c} = 0.
\end{align}
This completes the proof.
\end{proof}

\subsection{Proof of Theorem~\ref{th4}}\label{ApF}

\begin{proof}
In \cite{super_pilot}, it has been proved that the minimal $\mathbb{E}\{\|\hat{\bm{\alpha}} - \bm{\alpha}\|^2\}$ can be achieved when the column vectors of the matrix $\mathbf{\Psi}_{\rm{p}} $ in (\ref{r2}) are orthogonal with each other.
We now prove that a pilot sequence $\mathbf{x}_{\rm{p}}$ with an ideal ambiguity function always ensures that any two column vectors of $\mathbf{\Psi}_{\rm{p}} $ are mutually orthogonal.
The inner product between the $i$-th and $j$-th columns of $\mathbf{\Psi}_{\rm{p}}$ for $i\neq j$ can be expressed as
\begin{align}\label{eq-Th4} 
  & \mathbf{x}_{\rm{p}}^{\rm H} \mathbf{\Phi}_i^{\rm H} \mathbf{\Phi}_j \mathbf{x}_{\rm{p}} = \mathbf{x}_{\rm{p}} \mathbf{A} \left(\mathbf{\Gamma}_{\text{CPP}_i} \mathbf{\Pi}^{\tau_i} \mathbf{\Delta}_{\nu_i}\right)^{\rm H} \mathbf{\Gamma}_{\text{CPP}_j} \mathbf{\Pi}^{\tau_j} \mathbf{\Delta}_{\nu_j} \mathbf{A}^{\rm H} \mathbf{x}_{\rm{p}} \nonumber \\
  & \hspace{5mm}=  \mathbf{x}_{\rm{p}} \mathbf{A} \mathbf{\Delta}_{-\nu_i} \mathbf{\Pi}^{N_c-\tau_i} \mathbf{\Gamma}_{\text{CPP}_i}^{\rm H} \mathbf{\Gamma}_{\text{CPP}_j} \mathbf{\Pi}^{\tau_j} \mathbf{\Delta}_{\nu_j} \mathbf{A}^{\rm H} \mathbf{x}_{\rm{p}}. 
\end{align}
When $2c_1N_c$ is an integer, $\mathbf{\Gamma}_{{\text{CPP}}_i}\! =\! \mathbf{I}_{N_c}$. According to the definitions of $\mathbf{\Pi}^{\tau_i}$ and $\mathbf{\Delta}_{\nu_i}$, $\mathbf{\Pi}^{\tau_i} \mathbf{\Delta}_{\nu_i}\! =\! e^{-\textsf{j}2\pi \nu_i \tau_i/N} \mathbf{\Delta}_{\nu_i} \mathbf{\Pi}^{\tau_i}$. Based on these properties, $\mathbf{x}_{\rm{p}}^{\rm H} \mathbf{\Phi}_i^{\rm H} \mathbf{\Phi}_j \mathbf{x}_{\rm{p}}$ can be simplified as 
\begin{align}\label{inner1} 
  \mathbf{x}_{\rm{p}}^{\rm H} \mathbf{\Phi}_i^{\rm H} \mathbf{\Phi}_j \mathbf{x}_{\rm{p}} &= \mathbf{x}_{\rm{p}} \mathbf{A} \mathbf{\Delta}_{-\nu_i} \mathbf{\Pi}^{\tau_j-\tau_i} \mathbf{\Delta}_{\nu_j} \mathbf{A}^{\rm H} \mathbf{x}_{\rm{p}} \nonumber \\
  & \hspace*{-10mm}= e^{-\textsf{j}2\pi \nu_i(\tau_j-\tau_i)/N} \mathbf{x}_{\rm{p}} \mathbf{A} \mathbf{\Delta}_{\nu_j-\nu_i}  \mathbf{\Pi}^{\tau_j-\tau_i} \mathbf{A}^{\rm H} \mathbf{x}_{\rm{p}} .
\end{align}
(\ref{inner1}) can be interpreted as the cross-correlation between the pilot sequence with its corresponding echo obtained with a time delay of $\hat{\tau} = \tau_j\!-\!\tau_i$ and Doppler shift of $\hat{\nu} = \tau_j\!-\!\tau_i$, which is aligned with the ambiguity function definition. Therefore, (\ref{inner1}) can be equivalently rewritten as 
\begin{align}\label{inner2} 
  \mathbf{x}_{\rm{p}}^{\rm H} \mathbf{\Phi}_i^{\rm H} \mathbf{\Phi}_j \mathbf{x}_{\rm{p}} =& e^{-\textsf{j}2\pi \nu_i\hat{\tau}/N} \chi_{\rm{p}}(\hat{\tau},\hat{\nu}).
\end{align}
Since $(\tau_i,\nu_i),(\tau_j,\nu_j)\! \in\! \Omega_{\rm{R}}$, we have $(\hat{\tau},\hat{\nu})\! \in\! \Omega_{\rm{A}}$. Furthermore, $(\hat{\tau},\hat{\nu})\! \neq\! (0,0)$ for $i\! \neq\! j$. According to the definition of ideal ambiguity function, $\mathbf{x}_{\rm{p}}^{\rm H} \mathbf{\Phi}_i^{\rm H} \mathbf{\Phi}_j \mathbf{x}_{\rm{p}}\! = \! 0$.
 \end{proof}

\subsection{Proof of Theorem~\ref{T5}}\label{ApG}

\begin{proof}
The ambiguity function of the proposed pilot sequence can be expressed as
\begin{align}\label{e44} 
	\chi_{\rm{p}}(\tau,\nu) =& C_2 \sum_{k_1=0}^{N_{\rm{p}}} \sum_{k_2=0}^{N_{\rm{p}}} x_p^{*}[k_1Q] x_p[k_2Q] e^{-\textsf{j}2\pi\frac{k_2 Q\tau}{N_c}} \nonumber \\ 
  &\times S(k_1Q,k_2Q,\tau,\nu) .
\end{align}
According to (\ref{s}), $S(k_1Q,k_2Q,\tau,\nu)$ is non-zero if and only if $\langle k_2Q-k_1Q\rangle_{N_c}\! =\! \text{loc}^{(\tau,\nu)}$. Based on this property, we discuss the value of $\chi_{\rm{p}}(\tau,\nu)$ in the following two cases. 

\emph{\textbf{Case 1}}:
When $Q$ cannot divide $\text{loc}^{({\tau},{\nu})}$, there do not exist values of $k_1$ and $k_2$ that satisfy equation $\langle k_2Q-k_1Q\rangle_{N_c}=\text{loc}^{( {\tau}, {\nu})}$, since $\langle k_2Q-k_1Q\rangle_{N_c}$ is always a multiple of $Q$.
Therefore, $S(k_1Q,k_2Q, {\tau}, {\nu}) $ is always zero and $\chi_{\rm{p}}( {\tau}, {\nu}) = 0$.

\emph{\textbf{Case 2}}:
When $Q$ can divide $\text{loc}^{({\tau},{\nu})}$, there exists an integer $\hat{k}$ satisfying $\text{loc}^{( {\tau}, {\nu})} = \hat{k}Q$. 
By recalling the expressions of $\text{loc}^{( {\tau}, {\nu})}$ and $Q$, we obtain
\begin{align}\label{e45} 
  2 c_1 N_c 2^r \hat{k} = 2 c_1 \tau N_c - \nu.
\end{align}
Based on the multipath separation ability of AFDM in the DAFT domain, i.e., if $\tau_1\! \neq\! \tau_2$ or $\nu_1\! \neq\! \nu_2$, $\text{loc}^{(\tau_1,\nu_1)}\! \neq\! \text{loc}^{(\tau_2,\nu_2)}$,
the only solution of (\ref{e45}) is $( {\tau}, {\nu})\! =\! ({2^r}\hat{k},0)$.
By substituting this solution into the (\ref{e44}), the expression is simplified as
\begin{align}\label{e46} 
	\chi_{\rm{p}}({\tau}, {\nu}) =& \chi_{\rm{p}}({2^r}\hat{k},0) =  C_3 \sum_{k_1=0}^{{N_{\rm{p}}}}  x_p^{*}[k_1Q] x_p[\langle k_1Q+\hat{k}Q\rangle_{N_c}]  \nonumber \\
  &  \times  e^{-\textsf{j}2\pi\frac{(k_1Q+\hat{k}Q) {2^r}\hat{k}}{N_c}} e^{\textsf{j}2\pi c_2Q^2 (2k_1\hat{k})\!}  ,
\end{align}
where $ C_3\! =\! C_2N_c e^{\textsf{j}2\pi c_2Q^2\hat{k}^2}$. Using $x_{\rm{p}}[m]$ of (\ref{pilot2}) in (\ref{e46}) leads to
\begin{align}\label{e47} 
	\chi_{\rm{p}}({2^r}\hat{k},0) = C_4 \sum_{k_1=0}^{N_{\rm{p}}}  z^{*}[k_1] z[\langle k_1 +\hat{k}\rangle_{N_{\rm{p}}}] ,   
\end{align}
where $C_4\! \! =C_2N_c\frac{\sigma^2_{\rm{p}}}{N_{\rm{p}}}e^{-\textsf{j}2\pi \frac{\hat{k}^2 2^r Q}{2N_c}}$. Since $z[k]$ is a CAZAC sequence, we have $\chi_{\rm{p}}({2^r}\hat{k},0)\! =\! 0$ when $\hat{k}\! \neq\! 0$, which means $\chi_{\rm{p}}( {\tau}, {\nu})\! =\! 0$ when $({\tau}, {\nu})\! \neq\! (0,0)$. By combining the conclusions of Case~1 and Case~2, the value of $\chi_{\rm{p}}({\tau},{\nu})$ is non-zero only at $(\tau,\nu)\! =\! (0,0)$. Therefore, the proposed pilot design, which satisfies (\ref{pilot2}), can ensure an ideal ambiguity function.
\end{proof}

\bibliographystyle{IEEEtran}

\begin{thebibliography}{10}
\vspace{-2mm}
\bibitem{Liu_jsac_22} 
F.~Liu, \emph{et al.}, ``Integrated sensing and communications: Toward dual-functional wireless networks for 6G and beyond,'' \emph{IEEE J. Sel. Areas Commun.}, vol.~40, no.~6, pp.~1728--1767, Jun.~2022.

\bibitem{Cui_network_21} 
Y.~Cui, F.~Liu, X.~Jing, and J.~Mu, ``Integrating sensing and communications for ubiquitous IoT: Applications, trends, and challenges,'' \emph{IEEE Netw.}, vol.~35, no.~5, pp.~158--167, Sep./Oct.~2021.

\bibitem{Dong_twc_23} 
F.~Dong, \emph{et al.}, ``Sensing as a service in 6G perceptive networks: A unified framework for ISAC resource allocation,'' \emph{IEEE Trans. Wirel. Commun.}, vol.~22, no.~5, pp.~3522--3536, May 2023.

\bibitem{Gonz_proc_24} 
N.~Gonz\'{a}lez-Prelcic, \emph{et~al.}, ``The integrated sensing and communication revolution for 6G: Vision, techniques, and applications,'' \emph{Proc. IEEE}, vol.~112, no.~7, pp.~676--723, Jul.~2024.
 
\bibitem{Saad_Network_20} 
J.~Zhao, S.~Xue, K.~Cai, X.~Mu, Y.~Liu, and Y.~Zhu, ``Near-field integrated sensing and communications for secure UAV networks", \emph{arXiv preprint}, arXiv: 2502.01003, Feb.~2025.
 

\bibitem{Zhou_ojcoms_22} 
W.~Zhou, R.~Zhang, G.~Chen, and W.~Wu, ``Integrated sensing and communication waveform design: A survey,'' \emph{IEEE Open J. Commun. Soc.}, vol.~3, pp.~1930--1949, Oct.~2022.

\bibitem{Liu_spm_23} 
F.~Liu, \emph{et~al.}, ``Seventy years of radar and communications: The road from separation to integration,'' \emph{IEEE Signal Process. Mag.}, vol.~40, no.~5, pp.~106--121, Jul.~2023.

\bibitem{Li_tvt_24} 
S.~Li, \emph{et~al.}, ``Transmit beamforming design for ISAC with stacked intelligent metasurfaces,'' \emph{IEEE Trans. Veh. Technol.}, 2024, [early access]. DOI:10.1109/TVT.2024.3517709.

\bibitem{Fan_iwcmc_24} 
F. Zhang, \emph{et~al.}, ``Multicarrier waveform design for mmWave/THz integrated sensing and communication,'' in \emph{Proc. IWCMC 2024} (Ayia Napa, Cyprus), May~26-31, 2024, pp.~501--506.

\bibitem{Kaushik_csm_24} 
A. Kaushik, \emph{et~al.}, ``Toward integrated sensing and communications for 6G: Key enabling technologies, standardization, and challenges,'' \emph{IEEE Commun. Standards Mag.}, vol.~8, no.~2, pp.~52--59, Jun.~2024.

\bibitem{Zhang_icc_24} 
X.~Zhang, Z.~Lu, and J.~Wang, ``Sensing-aided CSI feedback with deep learning for massive MIMO systems,'' in \emph{Proc. ICC 2024} (Denver, CO, USA), Jun.~9-13, 2024, pp.~2282--2287.

\bibitem{Sturm_procieee_11} 
T.~Mao, J.~Chen, Q.~Wang, C.~Han, Z.~Wang, and G.~K.~Karagiannidis, "Waveform design for joint sensing and communications in millimeter-wave and low terahertz bands," \emph{IEEE Trans. Commun.}, vol.~70, no.~10, pp.~7023--7039, Oct.~2022.


\bibitem{Koivunen_spm_24} 
V.~Koivunen, \emph{et~al.}, ``Multicarrier ISAC: Advances in waveform design, signal processing, and learning under nonidealities,'' \emph{IEEE Signal Process. Mag.}, vol.~41, no.~5, pp.~17--30, Sep.~2024,

\bibitem{Fan_jsac_24} 
F.~Zhang, \emph{et~al.}, ``Cross-domain dual-functional OFDM waveform design for accurate sensing/positioning,'' \emph{IEEE J. Sel. Areas Commun.}, vol.~42, no.~9, pp.~2259--2274, Sep.~2024.

\bibitem{Shi_systemsJ_21} 
C.~Shi, \emph{et al.}, ``Joint optimization scheme for subcarrier selection and power allocation in multicarrier dual-function radar-communication system,'' \emph{IEEE Syst. J.}, vol.~15, no.~1, pp.~947--958, Mar.~2021.

\bibitem{Chen_cl_23} 
Y.~Chen, \emph{et al.}, ``Joint subcarrier and power allocation for integrated OFDM waveform in RadCom systems,'' \emph{IEEE Commun. Lett.}, vol.~27, no.~1, pp.~253--257, Jan.~2023.

\bibitem{Fan_wcnc_24} 
F.~Zhang, \emph{et~al.}, ``Cross-domain multicarrier waveform design for integrated sensing and communication,'' in \emph{Proc.~WCNC 2024} (Dubai, United Arab Emirates), Apr.~21-24, 2024, pp.~1--6.

\bibitem{Liu_twc_15} 
H.~Yin, X.~Wei, Y.~Tang, and K.~Yang, ``Diagonally reconstructed channel estimation for MIMO-AFDM with inter-doppler interference in doubly selective channels,'' \emph{IEEE Trans. Wireless Commun.}, vol.~23, no.~10, pp.~14066--14079, Oct.~2024.


\bibitem{OTFS1} 
L.~Gaudio, M.~Kobayashi, G.~Caire, and G.~Colavolpe, ``On the effectiveness of OTFS for joint radar parameter estimation and communication,'' \emph{IEEE Trans. Wireless Commun.}, vol.~19, no.~9, pp.~5951--5965, Sep.~2020.

\bibitem{OTFS2} 
M.~F.~Keskin, \emph{et~al.}, ``Integrated sensing and communications with MIMO-OTFS: ISI/ICI exploitation and delay-Doppler multiplexing,'' \emph{IEEE Trans. Wirel. Commun.}, vol.~23, no.~8, pp.~10229--10246, Aug.~2024.

\bibitem{OCDM1} 
X.~Ouyang and J.~Zhao, ``Orthogonal chirp division multiplexing,'' \emph{IEEE Trans. Commun.}, vol.~64, no.~9, pp.~3946--3957, Sep.~2016.

\bibitem{OCDM2} 
L.~G.~D.~Oliveira, M.~B.~Alabd, B.~Nuss, and T.~Zwick, ``An OCDM radar-communication system,'' in \emph{Proc. EuCAP 2020} (Copenhagen, Denmark), Mar.~15-20, 2020, pp.~1--5.

\bibitem{AFDM} 
A.~Bemani, N.~Ksairi, and M.~Kountouris, ``Affine frequency division multiplexing for next generation wireless communications,'' \emph{IEEE Trans. Wirel. Commun.}, vol.~22, no.~11, pp.~8214--8229, Nov.~2023.

\bibitem{afdm_isac0} 
H.~Bao, H.~Zhuang, Z.~Wang, and G.~Pang, ``Performance trade-off between communication and sensing based on AFDM parameter adjustment,'' in \emph{Proc.~PIMRC 2024} (Valencia, Spain), Sep.~2-5, 2024, pp.~1--6.

\bibitem{afdm_isac1} 
J.~Zhu, \emph{et~al.}, ``AFDM-based bistatic integrated sensing and communication in static scatterer environments,'' \emph{IEEE Wirel. Commun. Lett.}, vol.~13, no.~8, pp.~2245--2249, Aug.~2024.

\bibitem{afdm_isac2} 
Y.~Ni, Z.~Wang, P.~Yuan, and Q.~Huang, ``An AFDM-based integrated sensing and communications,'' in \emph{Proc.~ISWCS 2022} (Hangzhou, China), Oct.~19-22, 2022, pp.~1--6.

\bibitem{afdm_isac3} 
A.~Bemani, N.~Ksairi, and M.~Kountouris, ``Integrated sensing and communications with affine frequency division multiplexing,'' \emph{IEEE Wirel. Commun. Lett.}, vol.~13, no.~5, pp.~1255--1259, May~2024.

\bibitem{super_pilot} 
K.~Zheng, \emph{et~al.}, ``Channel estimation for AFDM with superimposed pilots,'' \emph{IEEE Trans. Veh. Technol.}, 2024, [early access], DOI:10.1109/TVT.2024.3469380.

\bibitem{STAR} 
C.~B.~Barneto, \emph{et al.}, ``Full-duplex OFDM radar with LTE and 5G NR waveforms: Challenges, solutions, and measurements,'' \emph{IEEE Trans. Microwave Theory Tech.}, vol.~67, no.~10, pp.~4042--4054, Oct.~2019.

\bibitem{afdm_signal}
H.~Yin, \emph{et al.}, ``Affine frequency division multiplexing: Extending OFDM for scenario-flexibility and resilience'', \emph{arXiv e-prints}, arXiv:2502.04735, Feb.~2025.

\bibitem{correlation} 
F.~Zhang, T.~Mao, and Z.~Wang, ``Doppler-resilient design of CAZAC sequences for mmWave/THz sensing applications,'' \emph{IEEE Trans. Veh. Technol.}, vol.~73, no.~2, pp.~2878--2883, Feb.~2024.

\bibitem{2d_fft} 
C.~Sturm, T.~Zwick, and W.~Wiesbeck, ``An OFDM system concept for joint radar and communications operations,'' in \emph{Proc. VTC Spring 2009} (Barcelona, Spain), Apr.~26-29, 2009, pp.~1--5.

\bibitem{music} 
T.~Ma, Y.~Xiao, X.~Lei, and M.~Xiao, ``Integrated sensing and communication for wireless extended reality (XR) with reconfigurable intelligent surface,'' \emph{IEEE J. Sel. Topics Signal Process.}, vol.~17, no.~5, pp.~980--994, Sep.~2023.

\bibitem{rdm} 
M.~Kronauge and H.~Rohling, ``Fast two-dimensional CFAR procedure,'' \emph{IEEE Trans. Aerosp. Electron. Syst.}, vol.~49, no.~3, pp.~1817--1823, Jul.~2013.

\bibitem{af} 
Z.~Ye, \emph{et al.}, ``Low ambiguity zone: Theoretical bounds and Doppler-resilient sequence design in integrated sensing and communication systems,'' \emph{IEEE J. Sel. Areas Commun.}, vol.~40, no.~6, pp.~1809--1822, Jun.~2022.

\bibitem{Popovic_tit_18} 
B.~M.~Popovi{\'c}, ``Optimum sets of interference-free sequences with zero autocorrelation zones,'' \emph{IEEE Trans. Inf. Theory}, vol.~64, no.~4, pp.~2876--2882, Apr.~2018.

\bibitem{zc} 
M.~H.~AlSharif, \emph{et al.}, ``Zadoff-Chu coded ultrasonic signal for accurate range estimation,'' in \emph{Proc. EUSIPCO 2017} (Kos island, Greece), Aug.~28-Sep.~2, 2017, pp.~1250--1254.

\bibitem{arxiv} 
Y.~Zhou, \emph{et al.}, ``GI-free pilot-aided channel estimation for affine frequency division multiplexing systems'', \emph{arXiv e-prints}, arXiv:2404.01088, Apr.~2024.

\end{thebibliography}

\end{document}